\documentclass[lettersize,journal]{IEEEtran}
\usepackage{amsmath,amsfonts}
\usepackage{algorithmic}
\usepackage{algorithm}
\usepackage{array}
\usepackage[caption=false,font=footnotesize,labelfont=rm,textfont=rm]{subfig}
\usepackage{textcomp}
\usepackage{stfloats}
\usepackage{url}
\usepackage{verbatim}
\usepackage{graphicx}
\usepackage{cite}
\usepackage{hyperref}
\usepackage{tabularray}
\usepackage{color}
\usepackage{booktabs}
\usepackage{amsmath}
\usepackage{amssymb}
\newtheorem{theorem}{Theorem} 
\newtheorem{lemma}{Lemma} 
\newtheorem{proposition}{Proposition}

\begin{document}

\title{Pareto-Optimal Rate-CRLB Tradeoff via Anchor Placement Optimization for UAV-Assisted 3D ISAC}

\author{Haoyu Jiang,~\IEEEmembership{Student Member,~IEEE,} Hengyou Kong, Xiaoli Xu,~\IEEEmembership{Member,~IEEE,} Yong Zeng,~\IEEEmembership{Fellow,~IEEE}
        % <-this % stops a space
% \thanks{This paper was produced by the IEEE Publication Technology Group. They are in Piscataway, NJ.}% <-this % stops a space
% \thanks{Manuscript received April 19, 2021; revised August 16, 2021.}

\thanks{ 
H. Jiang, H. Kong, X. Xu and Y. Zeng are with the National Mobile Communications Research Laboratory, Southeast University, Nanjing 210096, China. Y. Zeng is also with the Purple Mountain Laboratories, Nanjing 211111, China (e-mail: \{230258936, 230268151, xiaolixu, yong\_zeng\}@seu.edu.cn). (Corresponding author: Yong Zeng.)}

}

\maketitle

\begin{abstract}
This paper considers an unmanned aerial vehicle (UAV)-assisted 3D integrated sensing and communication (ISAC) system, where a UAV is deployed to communicate with a base station (BS) and simultaneously monitor a volume of interest (VoI). By exploiting the flexible placement of the UAV anchor node, we aim to characterize the Pareto-optimal tradeoff between the communication rate and sensing CRLB. For the special case of a singleton VoI, the exact Pareto-optimal UAV anchor placement set is shown to be the line segment connecting the BS and the sensing target. For a radially truncated spherical sector VoI, we derive geometric conditions under which all Pareto-optimal UAV locations are confined to an axial segment inside the inner radial boundary. For general VoIs, we prove the convexity of the regional worst-case CRLB in the inner region, derive a Pareto-optimal radius upper bound, and develop an efficient algorithm to find the Pareto-optimal UAV anchor node placement. Numerical results validate the analytical placement structures and demonstrate that our proposed design achieves significantly enlarged rate-CRLB region over benchmark schemes.
\end{abstract}


\section{Introduction}\label{sec:introduction}

% \IEEEPARstart{I}{ntegrated} sensing and communication (ISAC) has emerged as an important wireless paradigm that enables communication and environmental sensing to share spectrum, hardware, and signal processing resources \cite{liu_isac_jsac2022}. A fundamental issue in ISAC is the characterization of the tradeoff between communication and sensing performance \cite{xuxiaoli,qianglong_tutorial,Xiong_tradeoff_tit2023}. For target localization applications, the achievable communication rate and the Cram\'er-Rao lower bound (CRLB) of position estimation provide two basic performance measures: the former quantifies information transmission efficiency, while the latter characterizes the fundamental accuracy limit of localization. Accordingly, CRLB-based ISAC designs have been studied extensively, where communication and localization performance are balanced through transmit beamforming, covariance, or signaling design \cite{liu_crb_tsp2022,hua_crbrate_twc2024,ren_crbrate_twc2024,mao_cs_region_twc2024,Mao_isacnet_jsac2026}.

\IEEEPARstart{I}{ntegrated} sensing and communication (ISAC) seeks to use common hardware, spectrum, and transmitted signals for both information delivery and environmental sensing \cite{qianglong_tutorial,liu_isac_jsac2022,zhang_ISAC_jstsp}. By turning communication infrastructure into a source of situational information, ISAC can support applications such as intelligent transportation, industrial automation, environmental monitoring, and low-altitude surveillance \cite{cui_ISAC_ieeenetwork,zhang_ISAC_survey_tutorial,wei_ISAC_iotj,lu_ISAC_iotj}. The two functions, however, usually compete for transmit power, spatial degrees of freedom, bandwidth, and other system resources. Their achievable performances must therefore be considered jointly rather than optimized in isolation. This has motivated the study of communication-sensing tradeoffs, including information-theoretic characterizations and designs based on achievable rate, estimation accuracy, detection performance, and sensing mutual information \cite{xuxiaoli,Xiong_tradeoff_tit2023,hua_crbrate_twc2024,liu_crb_tsp2022,ren_crbrate_twc2024}. Among these metrics, the Cram\'er-Rao lower bound (CRLB) \cite{huizhi} is particularly suitable when the sensing task is target localization, because it quantifies the information available for estimating the target parameters rather than only the strength of the received echo.

% A common premise in many existing rate-CRLB analyses is that the sensing and communication geometry is prescribed. Under this premise, the propagation distances, observation directions, and geometric relationship among the transmitter (Tx), sensing receiver (Rx), and target are determined before the signal design is carried out \cite{Song_ris_tsp2023,Hua_transmit_beamform_tvt2023}. The remaining degrees of freedom therefore mainly reside in waveform, beamforming, covariance, or receiver processing. This premise is reasonable for conventional ISAC systems supported by fixed terrestrial infrastructure. However, it no longer applies when an unmanned aerial vehicle (UAV) participates in the sensing and communication process \cite{Yang_uav_number_twc2023, Pan_trajectory_tvt2024}. Owing to its controllable mobility, the UAV location itself can be optimized, thereby turning the spatial geometry into an additional design variable for ISAC performance optimization.

A substantial body of ISAC research characterizes this tradeoff through signal-domain design \cite{Xiong_tradeoff_tit2023,hua_crbrate_twc2024,liu_crb_tsp2022,Hua_transmit_beamform_tvt2023,Mao_isacnet_jsac2026,Hou_beamforming_jsac2024,Deng_beamforming_twc_2023}. For example, the rate-CRLB region has been studied for vector Gaussian channels by optimizing the transmit distribution \cite{Xiong_tradeoff_tit2023}; MIMO rate-CRLB boundaries have been obtained by optimizing the transmit covariance for point and extended targets \cite{hua_crbrate_twc2024}; and joint radar-communication beamforming has been designed to minimize sensing CRLBs subject to communication quality-of-service constraints \cite{liu_crb_tsp2022}. Related studies optimize waveform, beamforming, power allocation, and time-frequency resources under different sensing criteria. Although these formulations reveal how shared radio resources should be divided or reused, the locations of the transmitter, receiver, and sensing target are generally prescribed. This assumption is natural for terrestrial systems built on fixed infrastructure, but it treats propagation distances, viewing angles, and sensing geometry as environmental conditions rather than design variables. Conversely, emerging wireless networks increasingly involve network-connected autonomous agents, such as ground robots and unmanned aerial vehicles (UAVs), whose physical locations can be actively controlled \cite{zeng_UAVswarm_twc}. Their mobility introduces an additional design degree of freedom, since the locations and formations of these agents directly determine the communication and sensing channels. It therefore becomes possible to reshape the achievable communication-sensing tradeoff through physical-location optimization, beyond conventional signal-domain resource allocation \cite{zeng_capacity}.

The controllable mobility of UAVs introduces such a geometric degree of freedom into ISAC systems \cite{meng_uav_isac_mwc2024,Yong_UAV,yuxuan_magazine,Abdissa_deployment_precoder_iotj2024}. Existing UAV-assisted ISAC studies have exploited this mobility mainly through trajectory and maneuver optimization. In integrated periodic sensing and communication, the UAV trajectory has been jointly designed with user association, target-sensing selection, and transmit beamforming to maximize system throughput under periodic sensing requirements \cite{meng_uav_isac_twc_2023}. Adaptable ISAC designs optimize the UAV trajectory and communication/sensing beamforming over a mission interval while allowing sensing to be performed on demand \cite{Deng_beamforming_twc_2023}. Three-dimensional (3D) trajectory and resource allocation have also been considered for UAV-assisted Internet-of-Things systems using communication and radar-estimation-rate metrics \cite{Liu2024UAVIoT}. These works demonstrate the value of UAV mobility, but their primary concern is usually the time evolution of the UAV, accumulated communication performance, sensing schedules, or resource allocation under flight constraints \cite{jing_isac_sky_twc2024,jiang_uav_isac_cl2024,lyu_uav_isac_twc2023}. They do not directly answer a more fundamental question: for a cooperative UAV serving as an anchor node in bistatic ISAC systems, what are the UAV anchor placements for achieving the Pareto-optimal tradeoff between the communication rate and the localization accuracy to monitor the volume of interest (VoI)?

To address the above fundamental question, this paper investigates a UAV-assisted bistatic ISAC system. A cooperative UAV transmits an information-bearing ISAC waveform to a base station (BS) and assists the BS in monitoring a 3D VoI. Before a target is localized, its exact position is unavailable and only the VoI is assumed known. The UAV anchor node must therefore be placed to communicate with the BS while providing reliable localization for all potential targets in the VoI. We measure communication performance by the achievable rate and sensing performance by the worst-case 3D localization CRLB over the VoI, and characterize their Pareto tradeoff through the UAV placement optimization. The problem differs fundamentally from fixed-geometry optimization. The communication rate depends mainly on the BS-UAV distance, whereas the CRLB further depends on the UAV-target distances, target visibility relative to the BS array, and the delay-angle coupling induced by the BS-target-UAV geometry. Moreover, the target attaining the regional worst-case CRLB may change with the UAV location, and some outer-region placements yield an unbounded CRLB because a possible target becomes geometrically unresolvable.

The main contributions of this paper are summarized as follows.

\begin{itemize}
    \item We formulate an optimization problem to characterize the Pareto-optimal rate-CRLB tradeoff for 3D UAV-assisted bistatic ISAC. The achievable communication rate depends on the BS-UAV distance, whereas the regional sensing metric is defined as the worst-case CRLB over the VoI. We analytically characterize the Pareto-optimal UAV placement structures for representative VoIs. For a singleton VoI, the exact Pareto-optimal placement set is proved to be the line segment connecting the BS and the target, and the corresponding rate-CRLB Pareto boundary is obtained in closed form. For a radially truncated spherical sector VoI, we establish the inner-region convexity and derive a geometric condition under which every Pareto-optimal UAV placement is confined to an axial segment between the BS and the inner radial boundary of the VoI.
    
    \item We develop an algorithm to find the Pareto-optimal UAV anchor node placement. We prove that the regional worst-case CRLB is convex when the UAV is closer to the BS than every possible target and derive a radius upper bound for Pareto-optimal placements. Based on these properties, a two-stage procedure is proposed that combines convex radius-cap optimization in the inner region with fixed-radius directional search in the outer region.
    
    \item Numerical results are provided to validate the singleton and radially truncated spherical sector placement structures and reveal the behavior for a general off-axis rotated-ellipsoidal VoI. Over the evaluated Pareto-rate interval, the CRLB-based design attains a lower regional worst-case CRLB than sensing-SNR-based and VoI-center-based placements, which further shows that received echo strength or proximity to the VoI center does not necessarily provide favorable regional localization geometry and may result in substantially degraded worst-case localization performance. The optimized CRLB-based locations form a curved 3D path, demonstrating that the best direction may change with the communication requirement. Further experiments show that increasing the waveform root mean square (RMS) bandwidth yields diminishing localization returns once the angular and coupling terms become limiting, and that changing the UPA aspect ratio can alter both the Pareto boundary and the optimal UAV direction even when the total antenna number is fixed.
    
\end{itemize}

The remainder of this paper is organized as follows. Section~\ref{sec:system_model} presents the 3D UAV-assisted bistatic ISAC model, defines the communication and regional localization metrics, and formulates the rate-CRLB Pareto problem. Section~\ref{sec:pareto_analysis} analyzes the Pareto-optimal placements for the singleton and radially truncated spherical sector VoIs. Section~\ref{sec:general_roi_extension} develops the structural results and two-stage computation procedure for a general VoI. Section~\ref{sec:simulation_results} presents the numerical results. Finally, Section~\ref{sec:conclusion} concludes the paper.

% \emph{Notations:} Scalars, vectors, and matrices are denoted by italic letters, boldface lower-case letters, and boldface upper-case letters, respectively. $\mathbb{R}^{n}$ and $\mathbb{C}^{n}$ denote the spaces of $n$-dimensional real- and complex-valued vectors. For a vector $\mathbf{x}$, $\|\mathbf{x}\|$ is its Euclidean norm. $(\cdot)^{\rm T}$, $(\cdot)^{\rm H}$, $(\cdot)^{-1}$, and $\operatorname{Tr}(\cdot)$ denote the transpose, Hermitian transpose, inverse, and trace, respectively. Moreover, $\mathbf I_N$ denotes the $N$-dimensional identity matrix, $\mathcal{CN}(\boldsymbol\mu,\boldsymbol\Sigma)$ denotes a circularly symmetric complex Gaussian distribution with mean $\boldsymbol\mu$ and covariance matrix $\boldsymbol\Sigma$, and calligraphic letters denote sets.

\emph{Notations:} Scalars, vectors, and matrices are denoted by italic letters, boldface lower-case letters, and boldface upper-case letters, respectively. $\mathbb{R}^{n}$ and $\mathbb{C}^{n}$ denote the spaces of $n$-dimensional real- and complex-valued vectors. For a vector $\mathbf{x}$, $\|\mathbf{x}\|$ is its Euclidean norm. $(\cdot)^{\rm T}$, $(\cdot)^{\rm H}$, $(\cdot)^{-1}$, and $\operatorname{Tr}(\cdot)$ denote the transpose, Hermitian transpose, inverse, and trace, respectively. Moreover, $\mathbf I_N$ denotes the $N$-dimensional identity matrix.

\section{System Model}\label{sec:system_model}
\begin{figure}[htbp] 
        \centering \includegraphics[width=0.85\columnwidth]{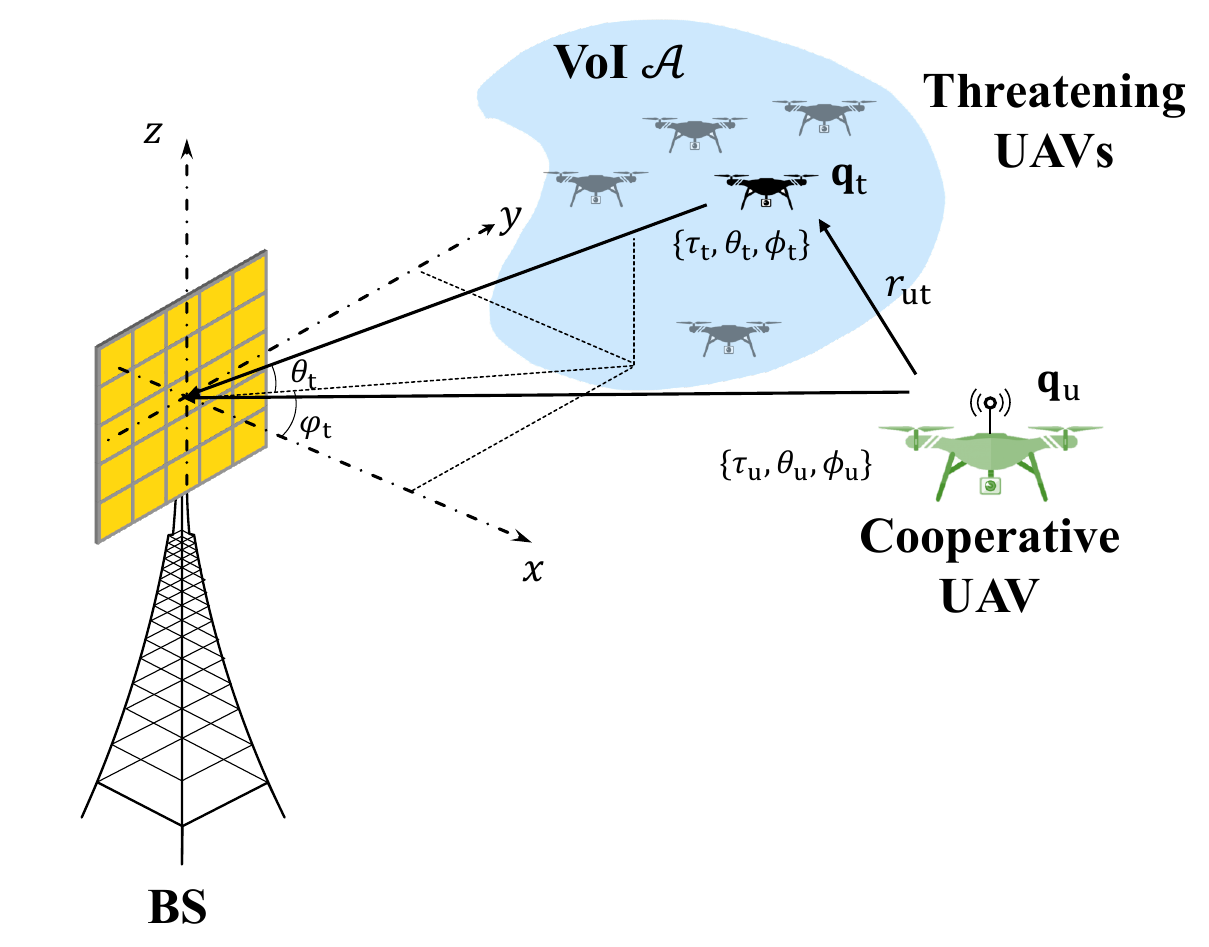}
        \caption{\label{fig:env} An illustration of UAV-assisted 3D ISAC system.}
\end{figure}

% As shown in Fig.~\ref{fig:env}, we consider a UAV-assisted bistatic ISAC system consisting of a BS, a cooperative UAV, and a 3D region of interest (RoI). The RoI represents a surveillance volume in which a target of interest, such as a threatening UAV, may appear at an unknown location. We consider one target at a time and assume that its radar cross section (RCS) is no smaller than a prescribed threshold $\kappa$. Before the target is localized, only its admissible region $\mathcal A$ and the RCS lower bound $\kappa$ are available to the system. Therefore, the cooperative UAV should be placed to communicate with the BS while ensuring reliable localization for every possible target location in $\mathcal A$ under the weakest admissible target return.

% The BS is located at the origin and is equipped with a uniform planar array (UPA) deployed on the $yz$-plane. The UPA has $N_y$ and $N_z$ antenna elements along the $y$- and $z$-axes, respectively, with a total of $N_{\rm R}=N_yN_z$ antennas. The inter-element spacing is $d=\lambda/2$, where $\lambda=c/f_c$, $f_c$ is the carrier frequency, and $c$ is the speed of light. A possible target location and the cooperative UAV location are denoted by $\mathbf q_{\rm t}=[x_{\rm t},y_{\rm t},z_{\rm t}]^{\rm T}\in\mathcal A$ and $\mathbf q_{\rm u}=[x_{\rm u},y_{\rm u},z_{\rm u}]^{\rm T}$, respectively. For $i\in\{{\rm u},{\rm t}\}$, the spherical-coordinate representation is

As shown in Fig.~\ref{fig:env}, we consider a 3D UAV-assisted bistatic ISAC system consisting of a BS, a cooperative UAV, and a 3D VoI. The cooperative UAV communicates with the BS while assisting the BS to monitor the VoI. Specifically, for any target (such as a threatening UAV) with radar cross section (RCS) no smaller than certain threshold entering the VoI, the system should be able to estimate the location of the target. We consider one typical sensing target and denote its RCS by $\kappa$, which is unknown before localization but satisfies $\kappa\geq\kappa_{\rm L}$, where $\kappa_{\rm L}$ is the RCS threshold. The cooperative UAV should be placed to communicate with the BS while ensuring reliable localization for every possible target location in $\mathcal A$ with RCS no smaller than $\kappa_L$.

The BS is located at the origin and is equipped with a uniform planar array (UPA) deployed on the $yz$-plane. The UPA has $N_y$ and $N_z$ antenna elements along the $y$- and $z$-axes, respectively, with a total of $N_{\rm R}=N_yN_z$ antennas. The inter-element spacing is $d=\lambda/2$, where $\lambda=c/f_c$, $f_c$ is the carrier frequency, and $c$ is the speed of light. The VoI to be monitored is denoted by $\mathcal A\subset\mathbb R^3$. Let $\mathbf q_{\rm t}=[x_{\rm t},y_{\rm t},z_{\rm t}]^{\rm T}\in\mathcal A$ denote one location within the VoI where the sensing target may appear, and $\mathbf q_{\rm u}=[x_{\rm u},y_{\rm u},z_{\rm u}]^{\rm T}$ is the location of the cooperative UAV which is to be optimized. For $i\in\{{\rm u},{\rm t}\}$, the spherical-coordinate representation is 
\begin{equation}
    \mathbf q_i=r_i
    \begin{bmatrix}
        \cos\theta_i\cos\phi_i,\
        \cos\theta_i\sin\phi_i,\
        \sin\theta_i
    \end{bmatrix},
\end{equation}
where $r_i=\|\mathbf q_i\|$, $\theta_i$ is the elevation angle, and $\phi_i$ is the azimuth angle. The UAV-target distance is
\begin{equation}
\begin{aligned}
    r_{\rm ut}
    =
    \Big(
    &r_{\rm u}^{2}+r_{\rm t}^{2}
    -2r_{\rm u}r_{\rm t}
    \big[
    \sin\theta_{\rm u}\sin\theta_{\rm t}\\
    &+
    \cos\theta_{\rm u}\cos\theta_{\rm t}
    \cos(\phi_{\rm u}-\phi_{\rm t})
    \big]
    \Big)^{1/2}.
\end{aligned}
\label{eq:r_ut_3d}
\end{equation}
The direct UAV-BS propagation delay and the bistatic sensing delay are $\tau_{\rm u}=r_{\rm u}/c$ and $\tau_{\rm t}=(r_{\rm t}+r_{\rm ut})/c$, respectively, where c is the speed of light.
Since the UPA is deployed on the $yz$-plane, the steering vector associated with direction $(\theta,\phi)$ is
\begin{equation}
    \mathbf a(\theta,\phi)
    =
    \mathbf a_y(\theta,\phi)\otimes\mathbf a_z(\theta),
\end{equation}
where
\begin{equation}
    \mathbf a_y(\theta,\phi)=
    \left[
    1,e^{-j\pi\cos\theta\sin\phi},\ldots,
    e^{-j\pi(N_y-1)\cos\theta\sin\phi}
    \right]^{\rm T},
\end{equation}
\begin{equation}
    \mathbf a_z(\theta)=
    \left[
    1,e^{-j\pi\sin\theta},\ldots,
    e^{-j\pi(N_z-1)\sin\theta}
    \right]^{\rm T}.
\end{equation}

\subsection{Signal Model}

The UAV transmits an ISAC signal such as OFDM waveform with observation interval $\mathcal T_{\rm obs}\triangleq[0,T_{\rm obs}]$. Let $s(t)$ denote the normalized complex baseband waveform satisfying $\mathbb{E}\big[|s(t)|^2\big] = 1.$
% \begin{equation}
%     \frac{1}{T_{\rm obs}}
%     \int_{\mathcal T_{\rm obs}}|s(t)|^2{\rm d}t=1.
%     \label{eq:normalized_waveform}
% \end{equation}
% \begin{equation}
%     \mathbb{E}\big[|s(t)|^2\big] = 1.
%     \label{eq:normalized_waveform}
% \end{equation}
The transmitted signal is denoted by $x(t)=\sqrt{P_{\rm t}}s(t)$, where $P_{\rm t}$ is the average UAV transmit power. The received signal at the BS is
\begin{equation}
\begin{aligned}
    \mathbf y(t)
    ={}&
    \sqrt{P_{\rm t}}\alpha_{\rm u}
    \mathbf a(\theta_{\rm u},\phi_{\rm u})
    s(t-\tau_{\rm u})\\
    &+
    \sqrt{P_{\rm t}}\alpha_{\rm t}
    \mathbf a(\theta_{\rm t},\phi_{\rm t})
    s(t-\tau_{\rm t})
    +\mathbf n(t),
    \quad t\in\mathcal T_{\rm obs},
\end{aligned}
\label{eq:received_signal}
\end{equation}
where $\alpha_{\rm u}$ and $\alpha_{\rm t}$ denote the direct line-of-sight (LoS) channel gain and the target-reflected channel gain, respectively, and $\mathbf n(t)$ is additive white Gaussian noise with two-sided power spectral density $N_0$. Under the Friis propagation model, the direct-path power gain is
\begin{equation}
    |\alpha_{\rm u}|^2=\zeta_0r_{\rm u}^{-2},
\end{equation}
where $\zeta_0\triangleq c^2/((4\pi)^2f_c^2)$ denotes the channel power gain at the reference distance of 1 meter.

For the target-reflected path, the bistatic radar equation gives
\begin{equation}
|\alpha_{\rm t}|^2
=
\frac{c^2\kappa}{(4\pi)^3f_c^2}
r_{\rm ut}^{-2}r_{\rm t}^{-2}
\geq
\xi_0r_{\rm ut}^{-2}r_{\rm t}^{-2},
\label{eq:target_path_gain}
\end{equation}
where
$\xi_0\triangleq c^2\kappa_{\rm L}/((4\pi)^3f_c^2)$.
The lower bound is attained when $\kappa=\kappa_{\rm L}$ and is adopted throughout the paper to characterize the guaranteed performance under the weakest target return.

\subsection{Communication Metric}

For the UAV uplink communication, the BS decodes the information-bearing waveform transmitted through the UAV-BS line of sight (LoS) link. Since the direct LoS path is typically much stronger than the target-reflected path, the communication channel is approximated by
\begin{equation}
    \mathbf h_{\rm c}
    =
    \alpha_{\rm u}\mathbf a(\theta_{\rm u},\phi_{\rm u}).
    \label{eq:communication_channel}
\end{equation}
With maximum-ratio combining at the BS, the communication SNR is
\begin{equation}
\begin{aligned}
\Upsilon_{\rm c}(\mathbf q_{\rm u})
&=
\frac{P_{\rm t}\|\mathbf h_{\rm c}\|^2}{N_0B}
=
\frac{P_{\rm t}N_{\rm R}\zeta_0}{N_0B r_{\rm u}^2}
=
\zeta_1r_{\rm u}^{-2},
\end{aligned}
\label{eq:communication_snr}
\end{equation}
where $B$ is the communication bandwidth and $\zeta_1\triangleq P_{\rm t}N_{\rm R}\zeta_0/(N_0B)$. The achievable communication spectral efficiency is
\begin{equation}
    \mathcal R(\mathbf q_{\rm u})
    =
    \zeta_2\log_2\left(1+\zeta_1r_{\rm u}^{-2}\right),
    \label{eq:communication_rate}
\end{equation}
where $0<\zeta_2\leq1$ accounts for possible waveform-dependent overhead.

\subsection{Sensing Metric}

For sensing, the direct UAV-BS component is assumed to be suppressed before target-parameter estimation using spatial or Doppler domain filters \cite{zihan}. The remaining target-reflected signal is
\begin{equation}
    \mathbf y_{\rm s}(t)
    =
    \sqrt{P_{\rm t}}\alpha_{\rm t}
    \mathbf a(\theta_{\rm t},\phi_{\rm t})
    s(t-\tau_{\rm t})
    +\mathbf n(t),
    \quad t\in\mathcal T_{\rm obs}.
\label{eq:sensing_signal}
\end{equation}
By assuming that the target state is static during the coherent processing interval (CPI) $\mathcal T_{\rm obs}$, sensing is performed by coherently combining the received signal over all information collected within the CPI. Hence, the accumulated SNR is 
\begin{equation}
\begin{aligned}
    \Upsilon_{\rm s}(\mathbf q_{\rm t},\mathbf q_{\rm u})
    &=
    \frac{|\alpha_{\rm t}|^2P_{\rm t}T_{\rm obs}\|\mathbf a(\theta_{\rm t},\phi_{\rm t})\|^2}{N_0}\\
    &=
    \frac{P_{\rm t}T_{\rm obs}N_{\rm R}\xi_0}
    {N_0r_{\rm ut}^2r_{\rm t}^2}
    =
    \frac{\xi_1}{r_{\rm ut}^2r_{\rm t}^2},
    \end{aligned}
    \label{eq:accumulated_sensing_snr}
    \end{equation}
where $\xi_1\triangleq P_{\rm t}T_{\rm obs}N_{\rm R}\xi_0/N_0$. Note that the SNR is related not only to the UAV placement $\mathbf{q}_{\rm u}$, but also the potential target location $\mathbf{q}_{\rm t}$ within the VoI. 

For target sensing, an important metric for potential target localization is CRLB. In our recent work \cite{CRLB}, we derived an explicit expression for the CRLB in bistatic sensing mode as a function of the UAV location $\mathbf{q}_{\rm u}$ and the potential target location $\mathbf{q}_{\rm t}$, given by

\begin{equation}
\begin{aligned}
    \mathcal C(\mathbf q_{\rm t},\mathbf q_{\rm u})
    ={}&
    \underbrace{
    \frac{c^2r_{\rm ut}^2r_{\rm t}^2}
    {\xi_\tau h^2}
    }_{\mathcal C_{\tau}}
    +
    \underbrace{
    \frac{r_{\rm ut}^2r_{\rm t}^4}{\ell_x^2}
    \left(
    \frac{1-\ell_z^2}{\xi_y}
    +
    \frac{1-\ell_y^2}{\xi_z}
    \right)
    }_{\mathcal C_{\rm ang}}
    \\
    &+
    \underbrace{
    \frac{4r_{\rm t}^2}{h^2\ell_x^2}
    \left(
    \frac{S_{xy}^2}{\xi_y}
    +
    \frac{S_{xz}^2}{\xi_z}
    \right)
    }_{\mathcal C_{\rm cpl}},
\end{aligned}
\label{eq:crlb_geometric_form}
\end{equation}
where
\begin{equation}
    \xi_y\triangleq\frac{\pi^2(N_y^2-1)\xi_1}{6},
    \qquad
    \xi_z\triangleq\frac{\pi^2(N_z^2-1)\xi_1}{6},
\end{equation}

\begin{equation}
    \xi_\tau\triangleq8\pi^2\beta_{\rm rms}^2\xi_1,
    \qquad
    [\ell_x,\ell_y,\ell_z]^{\rm T}
    \triangleq
    \frac{\mathbf q_{\rm t}}{r_{\rm t}},
    \label{eq:direction_cosines}
\end{equation}

\begin{equation}
\beta_{\rm rms}^2
=
\frac{\int (f-\bar f)^2|S(f)|^2{\rm d}f}
{\int |S(f)|^2{\rm d}f},
\qquad
\bar f=
\frac{\int f|S(f)|^2{\rm d}f}
{\int |S(f)|^2{\rm d}f},
\label{eq:rms_bandwidth}
\end{equation}

\begin{equation}
    h
    \triangleq
    1+\frac{\mathbf q_{\rm t}^{\rm T}(\mathbf q_{\rm t}-\mathbf q_{\rm u})}
    {r_{\rm t}r_{\rm ut}}
    =
    \frac{(r_{\rm t}+r_{\rm ut})^2-r_{\rm u}^2}
    {2r_{\rm t}r_{\rm ut}},
    \label{eq:h_geometric}
\end{equation}

\begin{equation}
    S_{xy}
    \triangleq
    \frac{1}{2}
    \left|
    \mathbf e_z^{\rm T}(\mathbf q_{\rm t}\times\mathbf q_{\rm u})
    \right|,
    \qquad
    S_{xz}
    \triangleq
    \frac{1}{2}
    \left|
    \mathbf e_y^{\rm T}(\mathbf q_{\rm t}\times\mathbf q_{\rm u})
    \right|.
    \label{eq:projected_areas}
\end{equation}

Here, $S(f)$ is the Fourier transform of $s(t)$. The three terms in \eqref{eq:crlb_geometric_form} separately characterize the delay, angular, and delay-angle coupling contributions to the position error. The delay term $\mathcal C_{\tau}$ decreases with the RMS bandwidth since $\xi_{\tau}\propto\beta_{\rm rms}^{2}$, while its geometric sensitivity is governed by $h$. The angular term $\mathcal C_{\rm ang}$ is determined by the two UPA apertures through $\xi_y$ and $\xi_z$, and is amplified by $\ell_x^{-2}$ as the target approaches the UPA endfire plane. The coupling term $\mathcal C_{\rm cpl}$ arises from the nonorthogonality between the delay and angular information, where $S_{xy}$ and $S_{xz}$ are the projected areas of the BS-target-UAV triangle on the $xy$- and $xz$-planes, respectively.
The parameters $P_{\rm t}$, $T_{\rm obs}$, $\kappa$, and $N_0$ scale all CRLB terms through $\xi_1$, whereas $\beta_{\rm rms}$, $N_y$, and $N_z$ change their relative weights and may therefore alter the optimal UAV placement. Interested readers may refer to \cite{CRLB} for more details on the CRLB derivation and implications. 

For a sensing VoI $\mathcal A$, the worst-case CRLB is defined as 
\begin{equation}
C_{\max}(\mathbf q_{\rm u})
\triangleq
\max_{\mathbf q_{\rm t}\in\mathcal A}
\mathcal C(\mathbf q_{\rm t},\mathbf q_{\rm u}).
\label{eq:regional_worst_crlb}
\end{equation}
It is observed from \eqref{eq:communication_rate}, \eqref{eq:crlb_geometric_form} and \eqref{eq:regional_worst_crlb} that both the communication performance in terms of spectral efficiency and the sensing performance in terms of the CRLB depend on where the cooperative UAV is placed. The communication rate depends only on the BS-UAV distance $r_{\rm u}$ and therefore favors placing the UAV close to the BS, whereas $\mathcal C(\mathbf q_{\rm t},\mathbf q_{\rm u})$ further depends on the UAV-target distances and the bistatic localization geometry over the entire VoI. Consequently, a communication-favorable placement may provide poor localization accuracy, and this motivates the UAV placement optimization for communication-sensing performance tradeoff.

\subsection{Rate-CRLB Pareto Boundary Formulation}
\label{subsec:pareto_formulation}

We aim to jointly improve the communication rate and regional localization accuracy through cooperative UAV placement. The achievable rate-CRLB region is defined as
\begin{equation}
    \Omega_{\rm C}
    \triangleq
    \bigcup_{\mathbf q_{\rm u}\in\mathcal Q}
    \left\{
    (\widehat R,\widehat C)
    \;\middle|\;
    \widehat R = \mathcal R(\mathbf q_{\rm u}),
    \;
    \widehat C = C_{\max}(\mathbf q_{\rm u})
    \right\},
\label{eq:rate_crlb_region}
\end{equation}
where $\mathcal Q$ denotes the feasible UAV placement set. The Pareto boundary of $\Omega_{\rm C}$ consists of the operating points for which neither the communication rate can be increased nor the worst-case CRLB can be reduced without degrading the other. The corresponding bi-objective optimization problem is formulated as
\begin{equation}
\text{(P-CRLB):}\quad
\underset{\mathbf q_{\rm u}\in\mathcal Q}{\operatorname{min}}
\quad
\mathbf f(\mathbf q_{\rm u})
\triangleq
\begin{bmatrix}
C_{\max}(\mathbf q_{\rm u})\\
-\mathcal R(\mathbf q_{\rm u})
\end{bmatrix},
\label{eq:Pcrlb}
\end{equation}
where the vector minimization is understood in the Pareto sense: $\mathbf q_{\rm u}^{\star}$ is Pareto-optimal if no feasible placement achieves both smaller $C_{\max}$ and $-\mathcal R$.

% \begin{equation}
%     \text{(P-CRLB):}\quad
%     \min_{\mathbf q_{\rm u}\in\mathcal Q}
%     \left\{
%     C_{\max}(\mathbf q_{\rm u}),
%     -\mathcal R(\mathbf q_{\rm u})
%     \right\}.
% \label{eq:Pcrlb}
% \end{equation}

% The complete Pareto boundary can be generated by solving
% \begin{equation}
% \begin{aligned}
%     \text{(P-CRLB-}\bar R\text{):}\quad
%     \min_{\mathbf q_{\rm u}\in\mathcal Q}\quad
%     &C_{\max}(\mathbf q_{\rm u})\\
%     \mathrm{s.t.}\quad
%     &\mathcal R(\mathbf q_{\rm u})\geq\bar R,
% \end{aligned}
% \label{eq:Pcrlb_parameterized}
% \end{equation}

Since $\mathcal R(\mathbf q_{\rm u})$ is strictly decreasing in $r_{\rm u}=\|\mathbf q_{\rm u}\|$, define the minimum and maximum feasible BS-UAV distances as
\begin{equation}
    r_{\mathcal Q}^{\min}\triangleq
    \min_{\mathbf q_{\rm u}\in\mathcal Q}\|\mathbf q_{\rm u}\|,
    \qquad
    r_{\mathcal Q}^{\max}\triangleq
    \max_{\mathbf q_{\rm u}\in\mathcal Q}\|\mathbf q_{\rm u}\|.
\label{eq:feasible_radius_range}
\end{equation}
The maximum and minimum feasible communication rates, denoted by $\underline R$ and $\overline R$ respectively, are therefore attained when $\mathbf{q}_{\rm u}$ achieves $r_{\mathcal Q}^{\min}$ and $r_{\mathcal Q}^{\max}$, respectively. Accordingly, the Pareto boundary candidates can be generated by varying $\tilde R$ and solving the following optimization problem
\begin{equation}
\begin{aligned}
    \text{(P-CRLB-$\tilde R$):}\quad
    \underset{\mathbf q_{\rm u}\in\mathcal Q}{\operatorname{min}}
    \quad & C_{\max}(\mathbf q_{\rm u})\\
    \operatorname{subject~to}\quad
    & \mathcal R(\mathbf q_{\rm u})\geq\tilde R,
\end{aligned}
\label{eq:P_CRLB_rate}
\end{equation}
where $\tilde R\in[\underline R,\overline R]$.

Before solving the problem (P-CRLB-$\tilde R$), we first consider the two extreme points of the achievable rate-CRLB region, which are communication- and localization-oriented placements. The communication-optimal placement is selected from the feasible locations nearest to the BS as
\begin{equation}
    \mathbf q_{\rm u}^{\rm com}
    =
    \underset{\substack{\mathbf q_{\rm u}\in\mathcal Q,\ \
    \|\mathbf q_{\rm u}\|=r_{\mathcal Q}^{\min}}}
    {\arg\min}
    \ C_{\max}(\mathbf q_{\rm u}).
\label{eq:communication_endpoint}
\end{equation}
The sensing-optimal placement is
\begin{equation}
    \mathbf q_{\rm u}^{\rm sen}
    =
    \underset{\mathbf q_{\rm u}\in\mathcal Q_{\rm C}^{\rm sen}}
    {\arg\min}
    \ \|\mathbf q_{\rm u}\|,
    \quad
    \text{where}\ \ \mathcal Q_{\rm C}^{\rm sen}
    \triangleq
    \underset{\mathbf q_{\rm u}\in\mathcal Q}{\arg\min}
    \ C_{\max}(\mathbf q_{\rm u}),
\label{eq:sensing_endpoint}
\end{equation}
The corresponding Pareto boundary endpoints are denoted by
\begin{equation}
\begin{aligned}
    \left(R^{\rm com},C^{\rm com}\right)
    &\triangleq
    \left(
    \mathcal R(\mathbf q_{\rm u}^{\rm com}),
    C_{\max}(\mathbf q_{\rm u}^{\rm com})
    \right),\\
    \qquad
    \left(R^{\rm sen},C^{\rm sen}\right)
    &\triangleq
    \left(
    \mathcal R(\mathbf q_{\rm u}^{\rm sen}),
    C_{\max}(\mathbf q_{\rm u}^{\rm sen})
    \right).
\label{eq:pareto_endpoints}
\end{aligned}
\end{equation}

\section{Pareto-Optimal Rate-CRLB Analysis for Representative VoIs}
\label{sec:pareto_analysis}

Since the worst-case sensing performance depends on the VoI, solving problem \eqref{eq:Pcrlb} requires a clear specification of the VoI. In this section, we consider some special VoIs that admit a closed-form expression for the rate-CRLB Pareto boundary, while the solution for a general VoI is deferred to Section \ref{sec:general_roi_extension}.

\subsection{Singleton VoI}
\label{subsec:singleton_roi}

We first consider the special case where the VoI reduces to a prescribed target location, i.e.,
\begin{equation}
    \mathcal A=\{\bar{\mathbf q}_{\rm t}\},
    \qquad
    \bar{\mathbf q}_{\rm t}=\bar r_{\rm t}\bar{\mathbf e}_{r},
    \qquad
    [\bar\ell_x,\bar\ell_y,\bar\ell_z]^{\rm T}=\bar{\mathbf e}_{r},
    \label{eq:singleton_roi}
\end{equation}
where $\bar r_{\rm t}=\|\bar{\mathbf q}_{\rm t}\|$ and $\bar\ell_x$ is the direction cosine of the target along the normal direction of the BS UPA. In this case, the regional worst-case CRLB reduces to $\mathcal C_{\max}(\mathbf q_{\rm u})=\mathcal C(\bar{\mathbf q}_{\rm t},\mathbf q_{\rm u})$. Since $\mathcal R(\mathbf q_{\rm u})$ is strictly decreasing in $r_{\rm u}$, problem \eqref{eq:Pcrlb} becomes
% \begin{equation}
%     \text{(P-CRLB-S):}\quad
%     \min_{\mathbf q_{\rm u}\in\mathcal Q}
%     \left\{
%     \mathcal C(\bar{\mathbf q}_{\rm t},\mathbf q_{\rm u}),
%     r_{\rm u}
%     \right\}.
%     \label{eq:P_singleton}
% \end{equation}

\begin{equation}
    \text{(P-CRLB-S):}\quad
    \min_{\mathbf q_{\rm u}\in\mathcal Q}
    \begin{bmatrix}
        C_{\max}(\mathbf q_{\rm u})\\
        r_{\rm u}
    \end{bmatrix}.
    \label{eq:P_singleton}
\end{equation}

Define the BS-target line segment as
\begin{equation}
    \mathcal L_{\rm t}
    \triangleq
    \left\{
    \epsilon\bar{\mathbf q}_{\rm t}\mid 0\leq\epsilon\leq1
    \right\}.
    \label{eq:singleton_line_segment}
\end{equation}
We assume that $\mathcal L_{\rm t}\subseteq\mathcal Q$ and that the target direction avoids the angular singularity of the UPA, i.e., $\bar\ell_x\neq0$. The Pareto-optimal placement set can then be characterized in closed form.

\begin{theorem}
\label{the:singleton_pareto}
For the singleton VoI $\mathcal A=\{\bar{\mathbf q}_{\rm t}\}$, if $\mathcal L_{\rm t}\subseteq\mathcal Q$ and $\bar\ell_x\neq0$, the Pareto-optimal UAV placement set of problem \eqref{eq:P_singleton} is
\begin{equation}
    \mathcal L_{\rm Pareto}^{\rm sing}
    =
    \left\{
    \mathbf q_{\rm u}^{\star}
    =
    \epsilon\bar{\mathbf q}_{\rm t}
    \mid
    0\leq\epsilon\leq1
    \right\}.
    \label{eq:singleton_pareto_set}
\end{equation}
Equivalently, the Pareto-optimal distances satisfy
\begin{equation}
    r_{\rm u}^{\star}
    =
    \bar r_{\rm t}-r_{\rm ut}^{\star},
    \qquad
    0\leq r_{\rm ut}^{\star}\leq\bar r_{\rm t}.
    \label{eq:singleton_pareto_distance}
\end{equation}
\end{theorem}

\begin{IEEEproof}
Please refer to Appendix~\ref{app:singleton_pareto}.
\end{IEEEproof}

Theorem~\ref{the:singleton_pareto} reduces the original 3D placement problem to a one-dimensional tradeoff along the BS-target line segment. At every Pareto-optimal location, the UAV lies between the BS and the target. Substituting \eqref{eq:singleton_pareto_distance} into \eqref{eq:crlb_geometric_form} gives
\begin{equation}
    \mathcal C^{\star}(r_{\rm ut})
    =
    \Gamma_{\rm t}r_{\rm ut}^{2},
    \qquad
    0\leq r_{\rm ut}\leq\bar r_{\rm t},
    \label{eq:singleton_crlb_tradeoff}
\end{equation}
where
\begin{equation}
\begin{aligned}
    \Gamma_{\rm t}
    \triangleq{}&
    \frac{c^{2}\bar r_{\rm t}^{2}}{4\xi_{\tau}}
    +
    \frac{\bar r_{\rm t}^{4}}{\bar\ell_x^{2}}
    \left(
    \frac{1-\bar\ell_z^{2}}{\xi_y}
    +
    \frac{1-\bar\ell_y^{2}}{\xi_z}
    \right)\\
    ={}&
    \frac{1}{\xi_1}
    \left[
    \frac{\lambda^2\bar r_{\rm t}^{2}}
    {32\pi^2\bar\beta^{\,2}}
    +
    \frac{6\bar r_{\rm t}^{4}}
    {\pi^2\bar\ell_x^{2}}
    \left(
    \frac{1-\bar\ell_z^{2}}{N_y^2-1}
    +
    \frac{1-\bar\ell_y^{2}}{N_z^2-1}
    \right)
    \right],
\end{aligned}
\label{eq:singleton_gamma}
\end{equation}
where $\bar\beta=\beta_{\rm rms}/f_c$. The corresponding communication rate is
\begin{equation}
    R^{\star}(r_{\rm ut})
    =
    \zeta_2\log_2
    \left(
    1+
    \frac{\zeta_1}
    {(\bar r_{\rm t}-r_{\rm ut})^{2}}
    \right).
    \label{eq:singleton_rate_tradeoff}
\end{equation}
Eliminating $r_{\rm ut}$ between \eqref{eq:singleton_crlb_tradeoff} and \eqref{eq:singleton_rate_tradeoff} yields the closed-form Pareto boundary
\begin{equation}
\begin{aligned}
    R^{\star}(\mathcal C^{\star})
    =
    \zeta_2\log_2
    \left(
    1+
    \frac{\zeta_1}
    {\left(
    \bar r_{\rm t}
    -
    \sqrt{\mathcal C^{\star}/\Gamma_{\rm t}}
    \right)^2}
    \right),
    \quad
    0\leq\mathcal C^{\star}\leq\Gamma_{\rm t}\bar r_{\rm t}^{2}.
\end{aligned}
\label{eq:singleton_pareto_boundary}
\end{equation}

% \begin{figure}[t]
%     \centering
%     \includegraphics[width=0.8\columnwidth]{fig/exp1_singleton_closed_form_only.pdf}
%     \caption{Closed-form rate-CRLB Pareto boundary for the singleton RoI.}
%     \label{fig:singleton_closed_form_boundary}
% \end{figure}

\begin{figure}[t]
    \centering
    \subfloat[Spatial Pareto placements.]{
        \includegraphics[width=0.7\columnwidth]{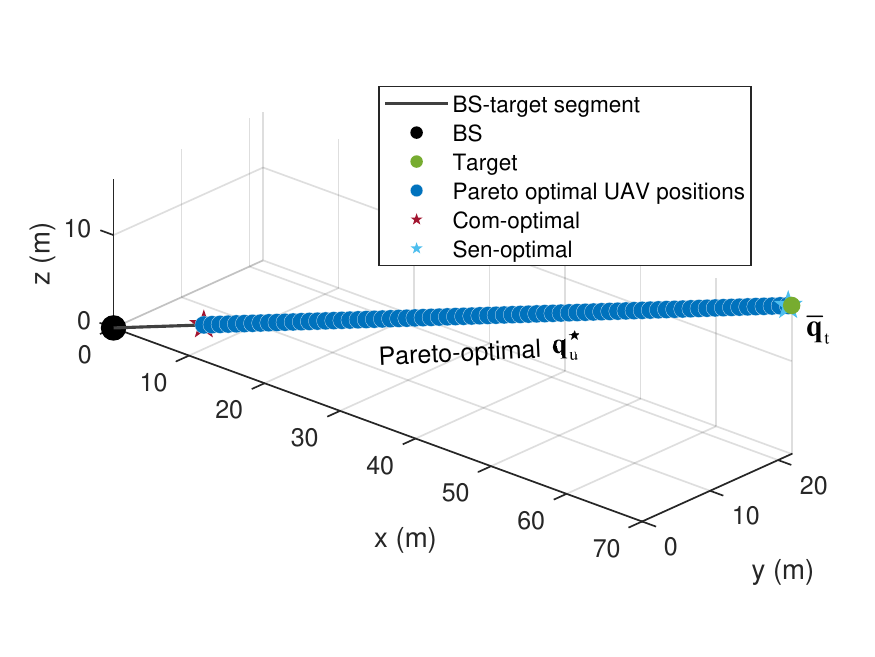}
        \label{fig:exp1_singleton_geometry}}\\
    \subfloat[Rate-CRLB Pareto boundary.]{
        \includegraphics[width=0.7\columnwidth]{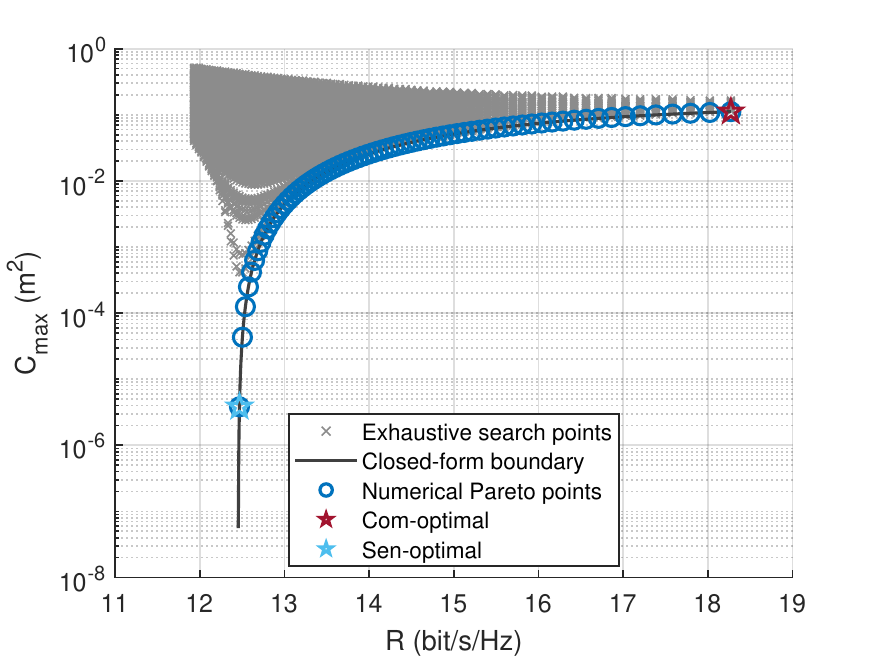}
        \label{fig:exp1_singleton_front}}
    \caption{Numerical illustration of Theorem~\ref{the:singleton_pareto}: (a) spatial Pareto-optimal UAV placements and (b) the corresponding closed-form and numerical rate-CRLB Pareto boundaries. The gray crosses denote all searched candidates, while the blue circles denote the resulting Pareto boundary.}
    \label{fig:exp1_singleton}
\end{figure}

Fig.~\ref{fig:exp1_singleton} numerically illustrates Theorem~\ref{the:singleton_pareto} and the closed-form boundary in \eqref{eq:singleton_pareto_boundary}. As shown in Fig.~\ref{fig:exp1_singleton_geometry}, all numerically identified Pareto-optimal UAV locations lie on the BS-target line segment. Moving the UAV from the BS-side communication-optimal endpoint toward the target reduces $r_{\rm ut}$ and hence the CRLB, at the cost of an increased $r_{\rm u}$ and a lower communication rate. Consequently, the numerical Pareto points in Fig.~\ref{fig:exp1_singleton_front} perfectly match with the closed-form boundary, while the gray crosses denote the rate-CRLB pairs of all searched candidates. The searching procedure is achieved by discretizing $\mathcal{Q}$, where $\mathcal{Q}$ is restricted by $d_{\min} \le r_{\rm u} \le d_{\max}$, and $d_{\min}$ and $d_{\max}$ denote the minimum and maximum BS-UAV distance. The boundary is strictly increasing with $C^\star$. Under the ideal point-target model, the sensing-optimal endpoint is attained in the limiting case $\mathbf q_{\rm u}\rightarrow\bar{\mathbf q}_{\rm t}$, yielding
\begin{equation}
\left(R^{\rm sen},C^{\rm sen}\right)
=
\left(
\zeta_2\log_2\left(1+\frac{\zeta_1}{\bar r_{\rm t}^{2}}\right),
0
\right).
\label{eq:singleton_sensing_endpoint}
\end{equation}
At the other end, with a practical minimum BS-UAV distance $d_{\min}>0$, the communication-optimal endpoint is
\begin{equation}
\left(R^{\rm com},C^{\rm com}\right)
=
\left(
\zeta_2\log_2\left(1+\frac{\zeta_1}{d_{\min}^{2}}\right),
\Gamma_{\rm t}\left(\bar r_{\rm t}-d_{\min}\right)^2
\right).
\label{eq:singleton_communication_endpoint}
\end{equation}
% A larger $\Gamma_{\rm t}$ increases the CRLB at every fixed rate and forces the UAV closer to the target to satisfy a prescribed localization requirement, thereby making the rate-CRLB conflict more pronounced. In contrast, $\zeta_1$ and $\zeta_2$ only change the rate associated with each placement along the same Pareto segment.

The parameter $\epsilon$ in \eqref{eq:singleton_pareto_set} is the normalized BS-UAV distance along the Pareto segment and satisfies $\epsilon = r_{\rm u}^{\star}/\bar r_{\rm t}=1-r_{\rm ut}^{\star}/\bar r_{\rm t}$, which can be interpreted as the geometric tradeoff variable. Increasing $\epsilon$ moves the UAV from the BS toward the target, thereby reducing the localization CRLB while decreasing the communication rate. In particular,
\begin{equation}
\begin{aligned}
    \mathcal C^{\star}
    &=
    \Gamma_{\rm t}r_{\rm ut}^{2}
    =
    \Gamma_{\rm t}\bar r_{\rm t}^{2}(1-\epsilon)^2,\\
    R^{\star}
    &=
    \zeta_2\log_2
    \left(
    1+\frac{\zeta_1}{\epsilon^2\bar r_{\rm t}^{2}}
    \right).
\end{aligned}
\label{eq:singleton_normalized_pareto}
\end{equation}

% The coefficient $\Gamma_{\rm t}$ represents the localization penalty per squared UAV-target separation. Its first component is determined by the normalized waveform bandwidth $\bar\beta$ and corresponds to the delay-information penalty, whereas its second component is determined by the array aperture $N_y,N_z$ and target direction $\bar{\mathbf e}_{r}$ and corresponds to the angular-information penalty. At a fixed normalized placement $\alpha$, the delay and angular contributions to $\mathcal C^{\star}$ scale as $\bar r_{\rm t}^{4}$ and $\bar r_{\rm t}^{6}$, respectively. Therefore, as the target distance increases, localization becomes increasingly limited by angular information, making the array aperture and target visibility more important than the waveform bandwidth. Moreover, the communication parameters $\zeta_1$ and $\zeta_2$ determine the UAV location associated with a prescribed rate, whereas the waveform, array, and sensing-power parameters determine the CRLB achieved at that location without changing the singleton Pareto-placement support.

The coefficient $\Gamma_{\rm t}$ represents the localization penalty per squared UAV-target separation $r_{\rm ut}$. The waveform and array parameters  do not change the geometric form of the singleton Pareto placement set, which always lies on the BS-target line segment. Instead, they determine how far the UAV must move along this segment to achieve a prescribed CRLB. A smaller normalized bandwidth $\bar\beta$, a smaller array aperture characterized by $N_y$ and $N_z$, or poorer target visibility characterized by $\bar\ell_x$ increases $\Gamma_{\rm t}$, forcing the UAV closer to the target and causing a larger communication-rate loss for the same localization requirement. Hence, unfavorable sensing conditions make the rate-CRLB tradeoff more pronounced, whereas increasing $\bar\beta$, $N_y$, or $N_z$ allows the UAV to remain closer to the BS. The communication parameters $\zeta_1$ and $\zeta_2$ similarly do not change the placement support, but only change the rate associated with each Pareto-optimal location.

In practice, if minimum BS-UAV separation of $d_{\min}>0$ is imposed, the feasible Pareto segment is truncated to
\begin{equation}
    d_{\min}/\bar r_{\rm t}
    \leq
    \epsilon
    \leq
    1.
    \label{eq:singleton_truncated_segment}
\end{equation}

\subsection{Radially Truncated Spherical Sector VoI}
\label{subsec:rtss}

\begin{figure}[t] 
        \centering \includegraphics[width=0.7\columnwidth]{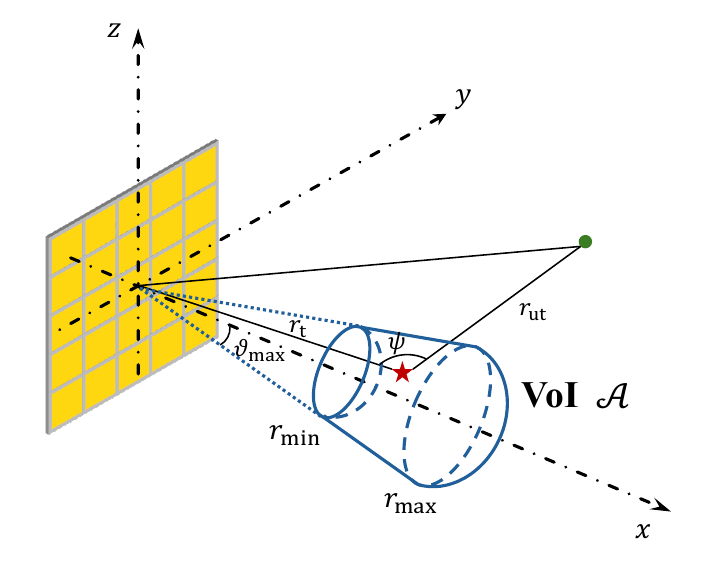}
        \caption{\label{fig:spherical_sector} An illustration of radially truncated spherical sector VoI.}
\end{figure}

As illustrated in Fig.~\ref{fig:spherical_sector}, we next consider a particular VoI $\mathcal A$, which is a radially truncated spherical sector formed by intersecting a circular cone with the spherical shell bounded by $r_{\min}$ and $r_{\max}$. The cone is centered on the positive $x$-axis with semi-aperture angle $\vartheta_{\max}$. The cooperative UAV is restricted to the front half-space $x>0$. To obtain a rotational symmetric CRLB property, we specialize the BS array to a square UPA, i.e.,
\begin{equation}
    N_y=N_z,
    \qquad
    \xi_y=\xi_z.
    \label{eq:square_upa}
\end{equation}

For any point $\mathbf q_i$, $i\in\{{\rm u},{\rm t}\}$, define its off-axis angle with respect to the positive $x$-axis as
\begin{equation}
    \vartheta_i
    \triangleq
    \arccos\left(\frac{\mathbf e_x^{\rm T}\mathbf q_i}{r_i}\right),
    \qquad
    \mathbf e_x\triangleq[1,0,0]^{\rm T}.
    \label{eq:off_axis_angle}
\end{equation}
The target VoI is characterized by the following constraints:
\begin{align}
    \mathrm{C_1}:&\quad r_{\min}\leq r_{\rm t}\leq r_{\max}, \label{eq:C1_cone}\\
    \mathrm{C_2}:&\quad 0\leq\vartheta_{\rm t}\leq\vartheta_{\max}, \label{eq:C2_cone}
\end{align}
where $r_{\min}$ and $r_{\max}$ specify the radial extent of the VoI, and $\vartheta_{\max}$ denotes the semi-aperture angle of the cone. We assume that $0\leq\vartheta_{\max}<\pi/2$. Accordingly, the radially truncated spherical sector VoI is defined as
\begin{equation}
    \mathcal A
    =
    \left\{
    \mathbf q_{\rm t}
    \mid
    \mathrm{C_1},\mathrm{C_2}
    \right\}.
    \label{eq:conical_roi}
\end{equation}

Under the square-UPA condition in \eqref{eq:square_upa}, we have $\xi_y=\xi_z$. Moreover, since $\ell_x^2+\ell_y^2+\ell_z^2=1$ and $\ell_x=\cos\vartheta_{\rm t}$, the geometric CRLB in \eqref{eq:crlb_geometric_form} reduces to
\begin{equation}
\begin{aligned}
    \mathcal C(\mathbf q_{\rm t},\mathbf q_{\rm u})
    ={}
    \frac{c^2r_{\rm ut}^2r_{\rm t}^2}
    {\xi_\tau h^2}
    +
    \frac{r_{\rm ut}^2r_{\rm t}^4
    \left(1+\ell_x^2\right)}
    {\xi_y\ell_x^2}
    +
    \frac{4r_{\rm t}^2
    \left(S_{xy}^2+S_{xz}^2\right)}
    {\xi_y h^2\ell_x^2} .
\end{aligned}
\label{eq:crlb_conical_geometric_form}
\end{equation}

The CRLB in \eqref{eq:crlb_conical_geometric_form} is invariant to a rotation of the target and UAV locations about the $x$-axis. Specifically, for any rotation matrix $\mathbf R_x(\omega)$ about the $x$-axis, we have
\begin{equation}
    \mathcal C\left(
    \mathbf R_x(\omega)\mathbf q_{\rm t},
    \mathbf R_x(\omega)\mathbf q_{\rm u}
    \right)
    =
    \mathcal C(\mathbf q_{\rm t},\mathbf q_{\rm u}).
    \label{eq:crlb_rotation_symmetry}
\end{equation}
Since $\mathcal A$ is also invariant to such rotations, the regional worst-case CRLB only depends on the UAV distance $r_{\rm u}$ and its off-axis angle $\vartheta_{\rm u}$. Therefore, without loss of optimality, we represent candidate UAV locations in the $xz$-plane as
% \begin{equation}
%     \mathbf q_{\rm u}
%     =
%     r_{\rm u}
%     \begin{bmatrix}
%         \cos\vartheta_{\rm u}\\
%         0\\
%         \sin\vartheta_{\rm u}
%     \end{bmatrix},
%     \qquad
%     r_{\rm u}\geq0,\quad
%     0\leq\vartheta_{\rm u}\leq\frac{\pi}{2}.
%     \label{eq:uav_reduced_representation}
% \end{equation}
\begin{equation}
    \mathbf q_{\rm u}
    =
    r_{\rm u}
    \begin{bmatrix}
        \cos\vartheta_{\rm u},
        0,
        \sin\vartheta_{\rm u}
    \end{bmatrix},
    \quad
    r_{\rm u}\geq0,\quad
    0\leq\vartheta_{\rm u}\leq\frac{\pi}{2}.
    \label{eq:uav_reduced_representation}
\end{equation}
Accordingly, the reduced UAV candidate domain is
\begin{equation}
    \mathcal Q
    \triangleq
    \left\{
    \mathbf q_{\rm u}
    \mid
    r_{\rm u}\geq0,\ 
    0\leq\vartheta_{\rm u}\leq\frac{\pi}{2}
    \right\}
    =
    \mathcal Q_1\cup\mathcal Q_2.
    \label{eq:uav_candidate_domain_cone}
\end{equation}
where
\begin{align}
    \mathcal Q_1
    &\triangleq
    \left\{
    \mathbf q_{\rm u}
    \mid
    0\leq r_{\rm u}<r_{\min},\
    0\leq\vartheta_{\rm u}\leq\frac{\pi}{2}
    \right\},
    \label{eq:Q1_cone}\\
    \mathcal Q_2
    &\triangleq
    \left\{
    \mathbf q_{\rm u}
    \mid
    r_{\rm u}\geq r_{\min},\
    0\leq\vartheta_{\rm u}\leq\frac{\pi}{2}
    \right\}.
    \label{eq:Q2_cone}
\end{align}
Here, $\mathcal Q_1$ contains UAV locations closer to the BS than every possible target point, whereas $\mathcal Q_2$ contains UAV locations outside $\mathcal Q_1$.

\begin{figure}[t]
    \centering
    \subfloat[Spatial Pareto placements.]{
        \includegraphics[width=0.7\columnwidth]{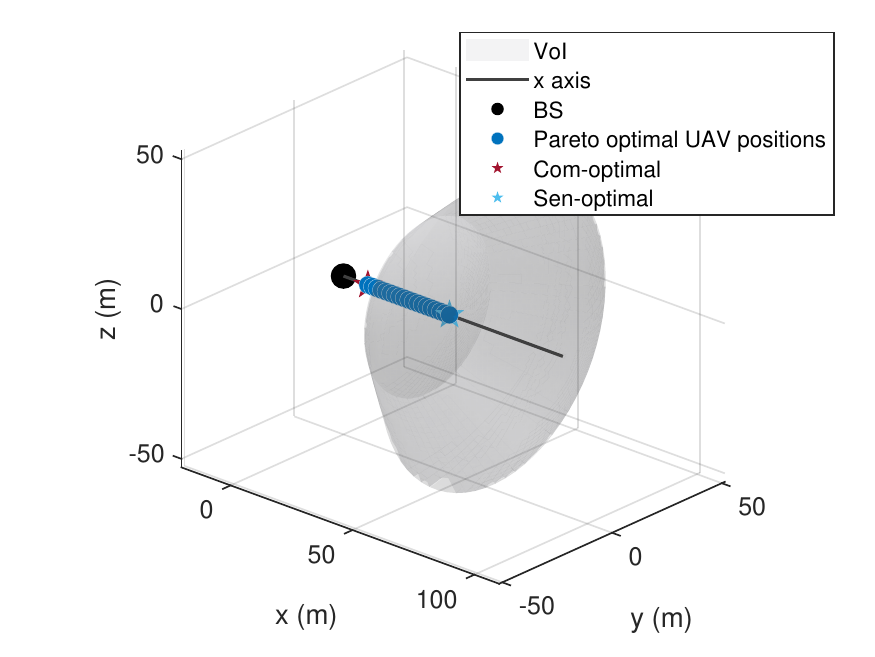}
        \label{fig:exp1_rtss_geometry}}\\
    \subfloat[Rate-CRLB Pareto boundary.]{
        \includegraphics[width=0.7\columnwidth]{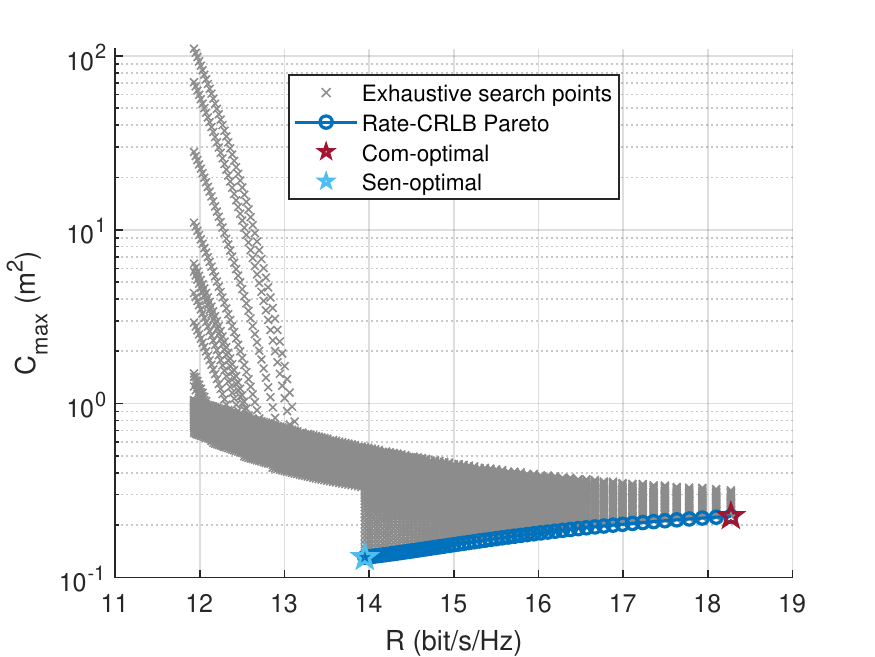}
        \label{fig:exp1_rtss_front}}
    \caption{Numerical illustration of the Pareto-placement structure in Theorem~\ref{the:conical_RoI_pareto}: (a) spatial Pareto-optimal UAV locations for the radially truncated spherical sector VoI and (b) the corresponding rate-CRLB boundary. The gray crosses denote all searched placement candidates.}
    \label{fig:exp1_rtss}
\end{figure}

% To avoid unbounded CRLB, we define the UAV forbidden set as
% \begin{equation}
%     \overline{\mathcal Q}
%     \triangleq
%     \left\{
%     \mathbf q_{\rm u}\in\mathcal Q_0
%     \mid
%     r_{\rm u}\geq r_{\min},\
%     0\leq\vartheta_{\rm u}\leq\vartheta_{\max}
%     \right\}.
%     \label{eq:uav_singularity_set_cone}
% \end{equation}
% The feasible UAV placement region for the CRLB-oriented design is thus given by
% \begin{equation}
%     \mathcal Q
%     =
%     \mathcal Q_0\setminus\overline{\mathcal Q}
%     =
%     \mathcal Q_1\cup\mathcal Q_2,
%     \label{eq:uav_feasible_region_cone}
% \end{equation}
% where
% \begin{align}
%     \mathcal Q_1
%     &\triangleq
%     \left\{
%     \mathbf q_{\rm u}
%     \mid
%     0\leq r_{\rm u}<r_{\min},\
%     0\leq\vartheta_{\rm u}\leq\frac{\pi}{2}
%     \right\},
%     \label{eq:Q1_cone}\\
%     \mathcal Q_2
%     &\triangleq
%     \left\{
%     \mathbf q_{\rm u}
%     \mid
%     r_{\rm u}\geq r_{\min},\
%     \vartheta_{\max}<\vartheta_{\rm u}\leq\frac{\pi}{2}
%     \right\}.
%     \label{eq:Q2_cone}
% \end{align}
% Here, $\mathcal Q_1$ contains UAV locations closer to the BS than every possible target point, whereas $\mathcal Q_2$ contains UAV locations outside the angular extent of the conical RoI. 

\begin{proposition}[Inner-region convexity]
\label{prop:crlb_inner_convexity}
For any fixed target $\mathbf q_{\rm t}\in\mathcal A$, $\mathcal C(\mathbf q_{\rm t},\mathbf q_{\rm u})$ is convex in $\mathbf q_{\rm u}$ over $\|\mathbf q_{\rm u}\|<\|\mathbf q_{\rm t}\|$. Consequently, $C_{\max}(\mathbf q_{\rm u})$ is convex over $\mathcal Q_1$.
\end{proposition}

\begin{IEEEproof}
Please refer to Appendix~\ref{app:crlb_inner_convexity}.
\end{IEEEproof}

\begin{lemma}
\label{lem:Q1_axis}
For the radially truncated spherical sector VoI in \eqref{eq:conical_roi}, suppose that the BS is equipped with an arbitrary rectangular UPA. If the UAV candidate region is restricted to $\mathcal Q_1$, every Pareto-optimal UAV location lies on the sector axis, i.e., $\mathbf q_{\rm u}^{\star}=r_{\rm u}^{\star}\mathbf e_x$ with $0\leq r_{\rm u}^{\star}<r_{\min}$.
\end{lemma}

\begin{IEEEproof}
    Please refer to Appendix \ref{app:Q1_axis}.
\end{IEEEproof}

\begin{lemma}\label{lem:Q2_dominated}
For the radially truncated spherical sector VoI in \eqref{eq:conical_roi} with the square-UPA condition in \eqref{eq:square_upa}, if
\begin{equation}
    \frac{r_{\min}}{r_{\max}}
    \geq
    \cos(2\vartheta_{\max}),
    \label{eq:Q2_domination_condition}
\end{equation}
then every UAV location in $\mathcal Q_2$ is dominated by a UAV location in $\mathcal Q_1$.
\end{lemma}

\begin{IEEEproof}
Please refer to Appendix~\ref{app:Q2_dominated}.
\end{IEEEproof}

\begin{theorem}
\label{the:conical_RoI_pareto}
For the radially truncated spherical sector VoI in \eqref{eq:conical_roi} and the feasible UAV placement region $\mathcal Q$, suppose that the BS is equipped with a square UPA satisfying \eqref{eq:square_upa}. If
\begin{equation}
    \frac{r_{\min}}{r_{\max}}
    \geq
    \cos(2\vartheta_{\max}),
    \label{eq:cone_pareto_subset_condition}
\end{equation}
then every Pareto-optimal UAV location satisfies
\begin{equation}
    \mathbf q_{\rm u}^{\star}
    =
    r_{\rm u}^{\star}\mathbf e_x,
    \qquad
    0\leq r_{\rm u}^{\star}<r_{\min}.
    \label{eq:cone_pareto_axis_subset}
\end{equation}
and the Pareto-optimal UAV positions can be characterized by
\begin{equation}
\mathcal L_{\rm Pareto}^{\rm RTSS}
\subseteq
\{r_{\rm u}\mathbf e_x:0\le r_{\rm u}<r_{\min}\}
\end{equation}
\end{theorem}

\begin{IEEEproof}
Please refer to Appendix~\ref{app:conical_RoI_pareto}.
\end{IEEEproof}

Fig.~\ref{fig:exp1_rtss} numerically illustrates the placement structure characterized by Theorem~\ref{the:conical_RoI_pareto}. As shown in Fig.~\ref{fig:exp1_rtss_geometry}, all nondominated UAV locations lie on the sector axis and satisfy $0\leq r_{\rm u}^{\star}<r_{\min}$, while only a subset of this axial segment is Pareto-optimal. Fig.~\ref{fig:exp1_rtss_front} shows the corresponding rate-CRLB boundary, where the gray crosses represent the exhaustive search points and the two endpoints correspond to the communication- and sensing-oriented placements. Compared with the singleton case, the radially truncated spherical sector Pareto boundary exhibits a considerably smaller sensing-performance variation under the considered setup. The regional worst-case operation limits the attainable localization gain and produces a flatter rate-CRLB boundary.

Theorem~\ref{the:conical_RoI_pareto} characterizes the support of the Pareto-optimal placements but does not guarantee every radius in \eqref{eq:cone_pareto_axis_subset} to be Pareto-optimal. The actual Pareto subset can be obtained by evaluating $\mathcal C_{\max}(r_{\rm u}\mathbf e_x)$ over $0\leq r_{\rm u}<r_{\min}$ and removing the dominated points. However, beyond the theorem conditions, the axial structure is no longer guaranteed. This fundamentally differs from the singleton case, where the Pareto-optimal set is only determined by the target parameter. For example, if $N_y\neq N_z$, the CRLB loses rotational symmetry, so the array dimensions $N_y$ and $N_z$, together with $\bar\beta$, may change the optimal UAV direction. These cases are handled by the general radius-wise computation developed in Section~\ref{sec:general_roi_extension}.

\section{Pareto-Optimal UAV Placement for General VoI}
\label{sec:general_roi_extension}

The preceding sections have derived closed-form Pareto-placement structures for representative VoIs. For a general VoI, such structures are generally unavailable because the worst-case target may vary with the UAV location. Nevertheless, the regional CRLB exhibits an inner-region convexity property, which enables a two-stage computation procedure that uses convex radius-cap optimization in the inner region and fixed-radius search in the outer region.

\begin{figure}[t] 
        \centering \includegraphics[width=0.7\columnwidth]{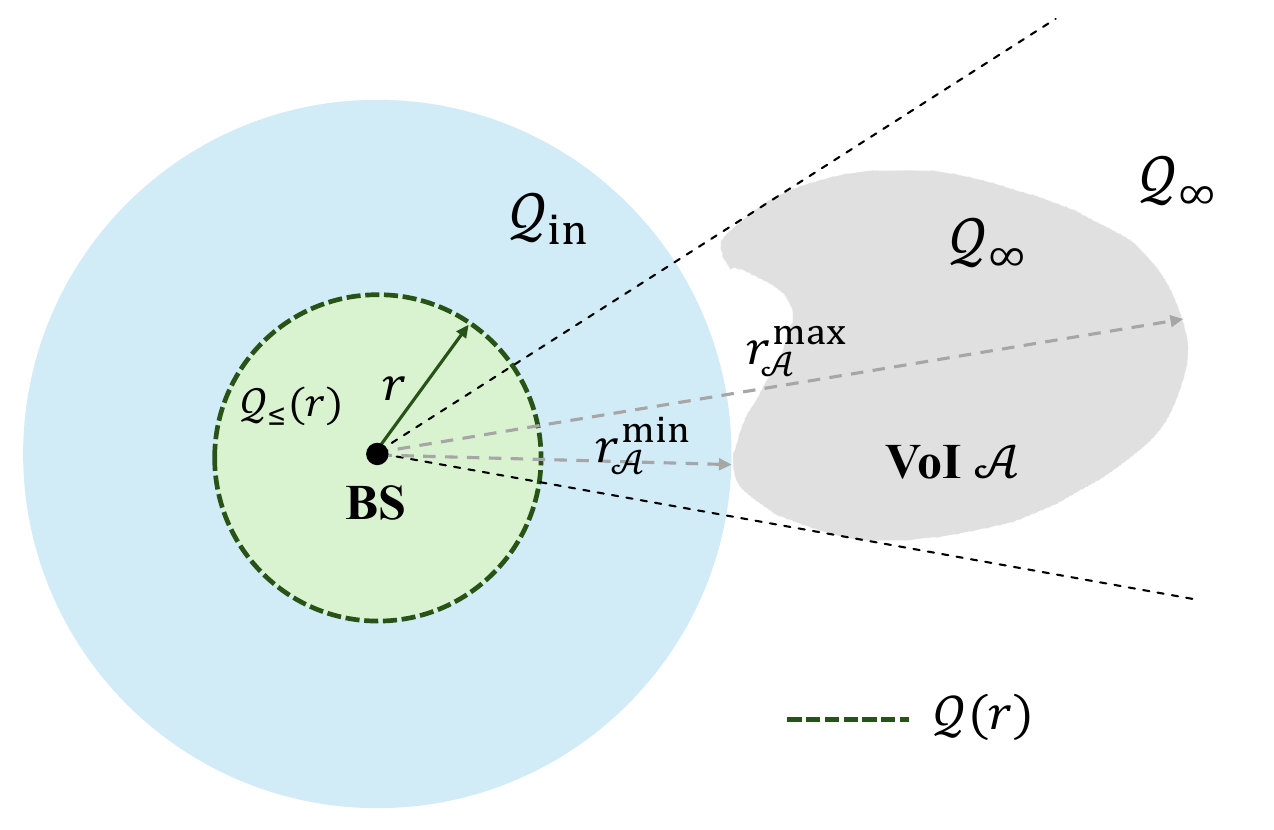}
        \caption{\label{fig:general_roi} An illustration of geometry symbols in Section \ref{sec:general_roi_extension}.}
\end{figure}

\subsection{General CRLB Structure}
\label{subsec:general_crlb_structure}

As shown in Fig. \ref{fig:general_roi}, for a VoI $\mathcal A$, define its minimum and maximum BS-target distances as
\begin{equation}
    r_{\mathcal A}^{\min}
    \triangleq
    \min_{\mathbf q_{\rm t}\in\mathcal A}
    \|\mathbf q_{\rm t}\|,
    \qquad
    r_{\mathcal A}^{\max}
    \triangleq
    \max_{\mathbf q_{\rm t}\in\mathcal A}
    \|\mathbf q_{\rm t}\|.
    \label{eq:roi_radial_extents}
\end{equation}
The inner region is defined by
\begin{equation}
    \mathcal{Q}_{\rm in}=\{ \mathbf{q}_{\rm u}\in\mathcal{Q} \mid \|\mathbf{q}_{\rm u}\| \le r_{\mathcal A}^{\min} \}.
    \label{eq:general_crlb_inner_region}
\end{equation}
We assume that $r_{\mathcal A}^{\min}>0$ and that all target directions avoid the angular singularity of the UPA, i.e., $\ell_x\neq0$ for all $\mathbf q_{\rm t}\in\mathcal A$. The result in Proposition \ref{prop:crlb_inner_convexity} holds for an arbitrary VoI and an arbitrary rectangular UPA. Hence, if $\mathcal Q$ is convex, minimizing the sampled regional CRLB under a radius cap $r<r_{\mathcal A}^{\min}$ is a convex problem. We next identify UAV locations yielding an infinite worst-case CRLB:
\begin{equation}
    \mathcal Q_{\infty}
    \triangleq
    \left\{
    \mathbf q_{\rm u}\in\mathcal Q
    \mid
    \exists\mathbf q_{\rm t}\in\mathcal A,\ 
    \mathbf q_{\rm u}=\alpha\mathbf q_{\rm t},\ 
    \alpha>1
    \right\}.
    \label{eq:general_crlb_infinite_set}
\end{equation}
The set $\mathcal Q_{\infty}$ does not intersect the inner region in \eqref{eq:general_crlb_inner_region}, and its exclusion generally makes the outer-region feasible set nonconvex. In addition, the Pareto search also admits a general radius bound below.

\begin{proposition}[Pareto-optimal UAV radius bound]
\label{prop:crlb_radius_bound}
Consider a finite-CRLB UAV location $\mathbf q_{\rm u}=R\mathbf e\in\mathcal Q$ with $\|\mathbf e\|=1$ and $R>r_{\mathcal A}^{\max}$. If there exists $\tilde r\in[r_{\mathcal A}^{\max},R)$ such that $\tilde r\mathbf e\in\mathcal Q$ and $C_{\max}(\tilde r\mathbf e)<\infty$, then $\mathbf q_{\rm u}=R\mathbf e$ cannot be rate-CRLB Pareto-optimal.
\end{proposition}

\begin{IEEEproof}
Please refer to Appendix~\ref{app:crlb_radius_bound}.
\end{IEEEproof}

Under the radial-feasibility condition in Proposition~\ref{prop:crlb_radius_bound}, the Pareto search can be restricted to
\begin{equation}
    0\leq r_{\rm u}\leq r_{\mathcal A}^{\max}.
    \label{eq:crlb_search_radius_interval}
\end{equation}

\subsection{Two-stage Computation Procedure}
\label{subsec:general_roi_computation}

For numerical evaluation, discretize the VoI as $\{\mathbf q_{{\rm t},n}\}_{n=1}^{N}$ and define
\begin{equation}
    C_{\max}^{(N)}(\mathbf q_{\rm u})
    \triangleq
    \max_{n=1,\ldots,N}
    \mathcal C(\mathbf q_{{\rm t},n},\mathbf q_{\rm u}).
    \label{eq:general_crlb_sampled_cmax}
\end{equation}
Since $\mathcal R(\mathbf q_{\rm u})$ depends only on $r_{\rm u}$ and decreases monotonically with it, the Pareto-optimal computation can be divided at $r_{\mathcal A}^{\min}$. Define the radius-cap and fixed-radius candidate sets as
\begin{equation}
\begin{aligned}
    \mathcal Q_{\leq}(r)
    &\triangleq
    \left\{
    \mathbf q_{\rm u}\in\mathcal Q
    \mid
    \|\mathbf q_{\rm u}\|\leq r
    \right\},\\
    \mathcal Q(r)
    &\triangleq
    \left\{
    \mathbf q_{\rm u}\in\mathcal Q
    \mid
    \|\mathbf q_{\rm u}\|=r
    \right\}.
\end{aligned}
\label{eq:general_radius_sets}
\end{equation}
The complete Pareto set is obtained by combining the candidates from the two stages below and removing repeated and dominated elements.

\begin{itemize}
    \item For $0\leq r<r_{\mathcal A}^{\min}$, the sensing-optimal placement under the radius cap is obtained from
    \begin{equation}
        \mathbf q_{\rm u}^{\rm in}(r)
        \in
        \arg\min_{\mathbf q_{\rm u}\in\mathcal Q_{\leq}(r)}
        C_{\max}^{(N)}(\mathbf q_{\rm u}).
        \label{eq:general_inner_cap_problem}
    \end{equation}
    If $\mathcal Q_{\leq}(r)$ is convex, \eqref{eq:general_inner_cap_problem} is a convex optimization problem. 

    \item For $r_{\mathcal A}^{\min}\leq r\leq r_{\mathcal A}^{\max}$, convexity is no longer guaranteed. The search is therefore performed at each fixed radius:
    \begin{equation}
        \mathbf q_{\rm u}^{\rm out}(r)
        \in
        \arg\min_{\mathbf q_{\rm u}\in
        \mathcal Q(r)\setminus\mathcal Q_{\infty}}
        C_{\max}^{(N)}(\mathbf q_{\rm u}).
        \label{eq:general_outer_radius_problem}
    \end{equation}
\end{itemize}

The whole procedure is summarized in Algorithm \ref{alg:general_roi_pareto}. We can conclude that UAV placement optimization separates avoidable geometric inefficiency from the intrinsic rate-localization tradeoff. At a fixed BS-UAV distance $r_{\rm u}$, all UAV directions achieve the same communication rate; hence, directional optimization can reduce the regional CRLB without any rate loss, and may even restore a finite CRLB by moving the UAV outside the CRLB-infinite region. Such geometric singularities cannot be removed merely by increasing the transmit power $P_{\rm t}$ or observation duration $T_{\rm obs}$. After the CRLB-optimal direction is selected at each radius, varying $r_{\rm u}$ reveals the unavoidable tradeoff between communication proximity and localization accuracy.

\begin{algorithm}[t]
\caption{Pareto-Optimal UAV Placement for General VoI}
\label{alg:general_roi_pareto}
\begin{algorithmic}[1]
\STATE \textbf{Input:} $\mathcal A$, $\mathcal Q$, and target samples $\{\mathbf q_{{\rm t},n}\}_{n=1}^{N}$
\STATE \textbf{Output:} Pareto-optimal UAV placement set $\mathcal P_{\rm C}$
\STATE $\mathcal P_{\rm C}\leftarrow\emptyset$
\FOR{each sampled $r\in[0,r_{\mathcal A}^{\min})$}
    \STATE $\widehat{\mathcal Q}_{\rm in}(r)\leftarrow\mathcal Q_{\leq}(r)$
    \IF{$\widehat{\mathcal Q}_{\rm in}(r)\neq\emptyset$}
        \STATE $\displaystyle \mathbf q_{\rm u}^{\rm in}(r)\in\arg\min_{\mathbf q_{\rm u}\in\widehat{\mathcal Q}_{\rm in}(r)}C_{\max}^{(N)}(\mathbf q_{\rm u})$
        \STATE Add $\bigl(\mathbf q_{\rm u}^{\rm in}(r),\mathcal R(\mathbf q_{\rm u}^{\rm in}(r)),C_{\max}^{(N)}(\mathbf q_{\rm u}^{\rm in}(r))\bigr)$ to $\mathcal P_{\rm C}$
    \ENDIF
\ENDFOR
\FOR{each sampled $r\in[r_{\mathcal A}^{\min},r_{\mathcal A}^{\max}]$}
    \STATE $\widehat{\mathcal Q}_{\rm out}(r)\leftarrow\mathcal Q(r)\setminus\mathcal Q_{\infty}$
    \IF{$\widehat{\mathcal Q}_{\rm out}(r)\neq\emptyset$}
        \STATE $\displaystyle \mathbf q_{\rm u}^{\rm out}(r)\in\arg\min_{\mathbf q_{\rm u}\in\widehat{\mathcal Q}_{\rm out}(r)}C_{\max}^{(N)}(\mathbf q_{\rm u})$
        \STATE Add $\bigl(\mathbf q_{\rm u}^{\rm out}(r),\mathcal R(\mathbf q_{\rm u}^{\rm out}(r)),C_{\max}^{(N)}(\mathbf q_{\rm u}^{\rm out}(r))\bigr)$ to $\mathcal P_{\rm C}$
    \ENDIF
\ENDFOR
\STATE Remove repeated and dominated elements from $\mathcal P_{\rm C}$
\RETURN $\mathcal P_{\rm C}$
\end{algorithmic}
\end{algorithm}

\section{Simulation Results}
\label{sec:simulation_results}
In this section, numerical simulations are conducted to evaluate the rate-CRLB tradeoff under both representative and general VoIs. Unless otherwise specified, the BS is located at the origin and is equipped with an $8\times 8$ UPA deployed on the $yz$-plane with half-wavelength antenna spacing. The common simulation parameters are summarized in Table~\ref{tab:simulation_parameters}. For each regional sensing scenario, all compared placement schemes are evaluated using the same target samples, and their Pareto placements are obtained using the search method described in Section~\ref{sec:general_roi_extension}.

\begin{table}[t]
\centering
\caption{Common Simulation Parameters}
\label{tab:simulation_parameters}
\renewcommand{\arraystretch}{1.1}
\setlength{\tabcolsep}{5pt}
\begin{tabular}{lcc}
\toprule
Parameter & Symbol & Value \\
\midrule
Carrier frequency & $f_c$ & $24$ GHz \\
Wavelength & $\lambda$ & $12.5$ mm \\
Transmit power & $P_{\rm t}$ & $0.5$ W \\
Observation duration & $T_{\rm obs}$ & $1$ ms \\
Noise power spectral density & $N_0$ & $-170$ dBm/Hz \\
Communication bandwidth & $B$ & $100$ MHz \\
RCS coefficient & $\kappa$ & $0.1$ m$^2$ \\
UPA configuration & $N_y\times N_z$ & $8\times8$ \\
Antenna spacing & $d$ & $\lambda/2=6.25$ mm \\
RMS bandwidth & $\beta_{\rm rms}$ & $20$ MHz \\
Normalized RMS bandwidth & $\bar{\beta}$ & $8.33\times10^{-4}$ \\
Rate pre-log factor & $\zeta_2$ & $1$ \\
Minimum BS-UAV distance & $d_{\min}$ & $10$ m \\
\bottomrule
\end{tabular}
\end{table}

\subsection{General VoI and Benchmark Comparison}
\label{subsec:sim_general_roi}

We consider a general off-axis VoI modeled as a rotated ellipsoid and compare the proposed CRLB-based placement with two benchmark schemes. For the SNR-based scheme, the UAV direction at each fixed BS-UAV distance is selected by maximizing the regional worst-case sensing SNR, whereas the resulting placement is evaluated using the regional worst-case CRLB. For the center-based scheme, the UAV is constrained to the ray from the BS toward the ellipsoid center. They are all evaluated over the Pareto-optimal rate interval generated by the CRLB criterion.

\begin{figure*}[t]
    \centering
    \subfloat[Rate-CRLB curves for different schemes.]{
        \includegraphics[width=0.28\textwidth]{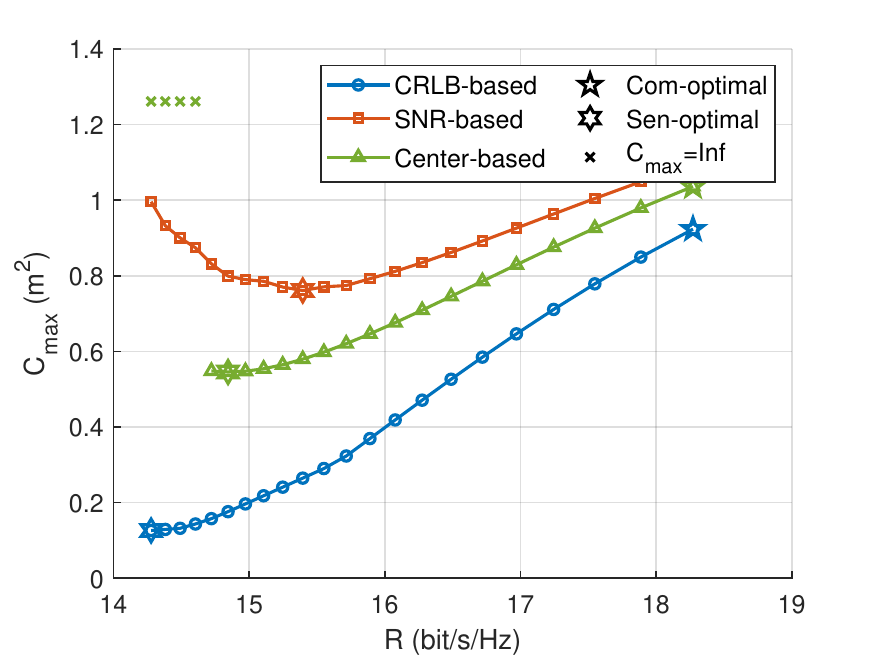}
        \label{fig:general_roi_front}}
    \hfill
    \subfloat[Pareto and benchmark UAV placements.]{
        \includegraphics[width=0.28\textwidth]{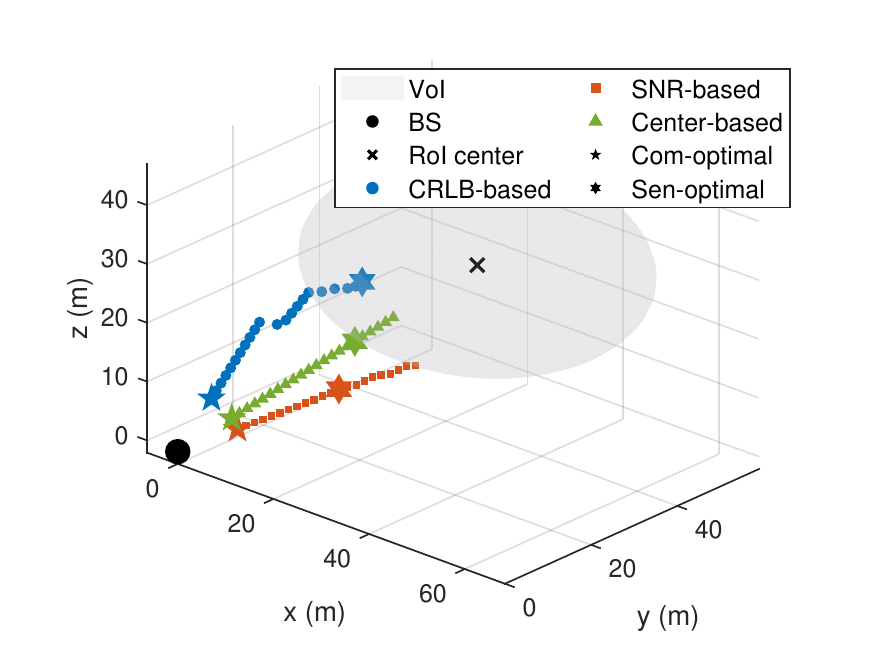}
        \label{fig:general_roi_placements}}
    \hfill
    \subfloat[UAV directions on the $\theta_{\rm u}$-$\phi_{\rm u}$ plane.]{
        \includegraphics[width=0.28\textwidth]{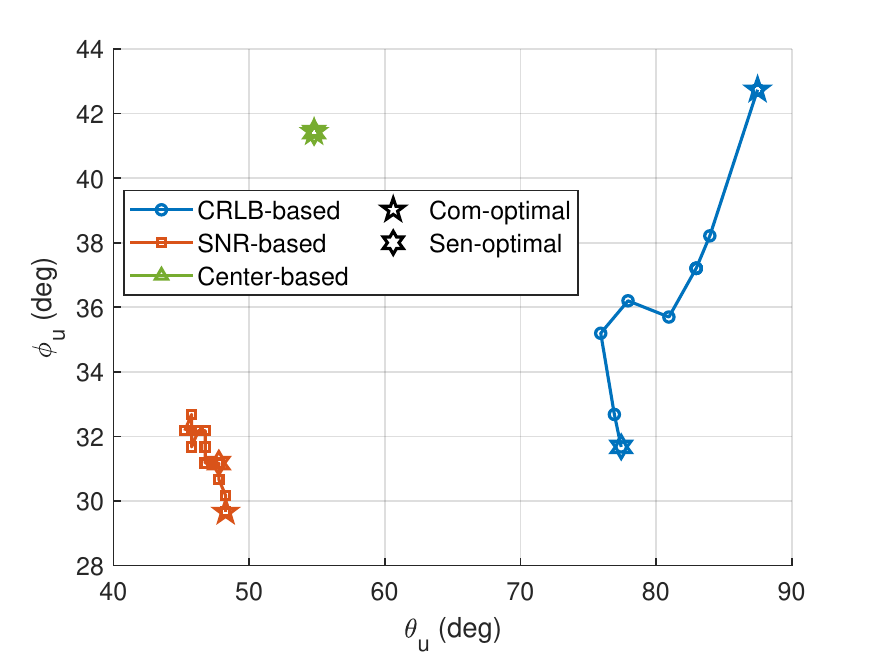}
        \label{fig:general_roi_angles}}
    \caption{Rate-CRLB performance and UAV placements for the general rotated-ellipsoidal VoI.}
    \label{fig:general_roi_comparison}
\end{figure*}

Fig.~\ref{fig:general_roi_front} shows that the proposed CRLB-based placement achieves the best CRLB throughout the entire rate range. Its curve exhibits the expected monotonic tradeoff. In contrast, the SNR-based curve is nonmonotonic and contains a dominated low-rate branch. Its minimum CRLB occurs at an intermediate rate. Therefore, a larger BS-UAV distance alone does not necessarily benefit localization. The result also demonstrates that the regional worst-case sensing SNR is not a reliable surrogate for localization accuracy. Although the SNR-based scheme optimizes the received sensing power at every radius, it does not necessarily provide sufficient delay-angle information for position estimation. The center-based scheme, where the UAV is constrained to the ray from the BS toward the ellipsoid center, provides better finite-CRLB performance than the SNR-based scheme over most of the rate range, and its sensing-optimal point is also attained at an intermediate radius. When the UAV moves farther along the center-directed ray, the placement enters the localization-singularity area inside the VoI, as indicated by the crosses at the top of Fig.~\ref{fig:general_roi_front}, the corresponding regional worst-case CRLB becomes unbounded. 

The spatial placements in Fig.~\ref{fig:general_roi_placements} further reveal the different design mechanisms. The center-based placements remain on the fixed ray connecting the BS and the ellipsoid center. The SNR-based placements also follow an almost straight trajectory, with only limited directional adjustment. By contrast, the CRLB-based placements form a distinctly curved three-dimensional path and deviate substantially from both benchmark directions. They select more elevated and laterally displaced positions to balance the bistatic localization geometry over the entire VoI rather than simply approaching its center or maximizing the received echo strength.

This distinction is more clearly observed in Fig.~\ref{fig:general_roi_angles}. Since the center-based direction is independent of the BS-UAV distance, all of its angular placements coincide at a single point, and its two endpoint markers overlap. The SNR-based directions remain within a narrow neighborhood, indicating that its rate variation is realized predominantly through radial motion. In contrast, the CRLB-based directions span a much wider range. The curved and locally nonmonotonic angular trajectory indicates that the localization-optimal direction varies with the communication requirement, likely due to changes in the active worst-case geometry and in the relative contributions of the delay, angular, and coupling information.

% These results demonstrate the radius-direction decomposition of the placement-induced tradeoff. Changing the UAV radius determines the communication-rate variation, whereas selecting the UAV direction determines the regional localization geometry at that rate. The proposed method exploits this directional degree of freedom to reduce the worst-case CRLB without incurring any additional communication-rate loss. Moreover, unlike the SNR- and center-based designs, it produces a genuine rate-CRLB Pareto boundary over the considered range and avoids regional localization singularities. The following experiments investigate how this boundary-achieving directional behavior changes with the waveform and array parameters.

\subsection{Impact of the RMS Bandwidth}
\label{subsec:sim_bandwidth}

We next investigate the impact of the waveform RMS bandwidth $\beta_{\rm rms}$ on the rate-CRLB boundary and the corresponding optimal UAV placements. The same rotated-ellipsoidal VoI and system configuration as in the preceding experiment are adopted, while $\beta_{\rm rms}$ is set to $1$, $2$, and $20$ MHz. 

\begin{figure*}[t]
    \centering
    \subfloat[Rate-CRLB Pareto boundaries.]{
        \includegraphics[width=0.28\textwidth]{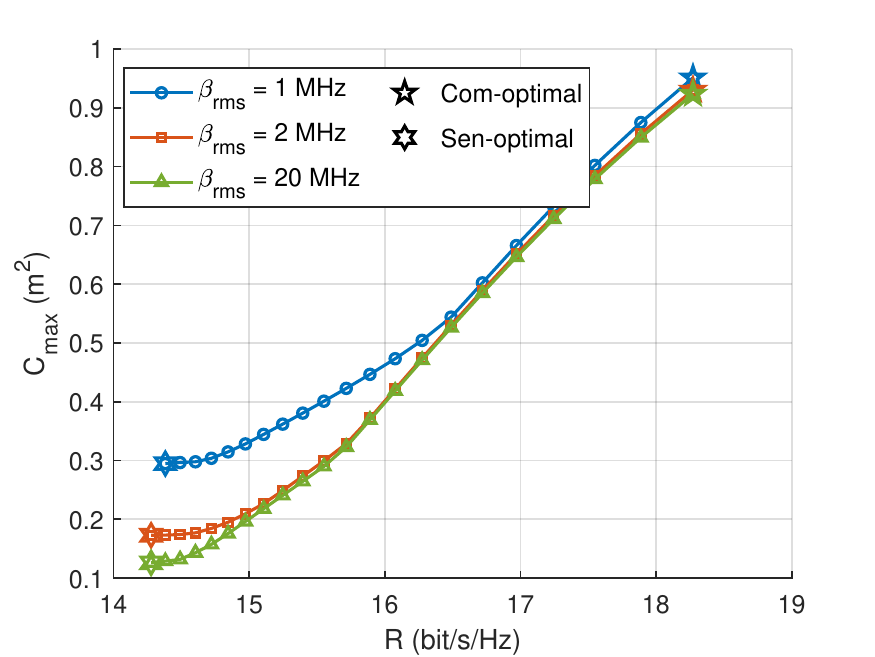}
        \label{fig:bandwidth_front}}
    \subfloat[Optimal UAV placements.]{
        \includegraphics[width=0.28\textwidth]{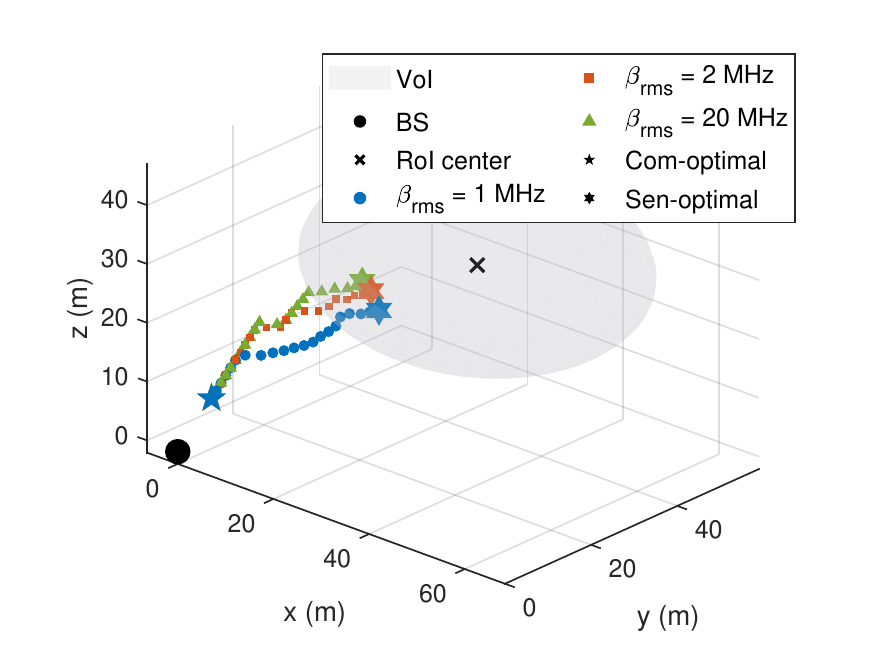}
        \label{fig:bandwidth_placement}}    
    \subfloat[Optimal UAV directions.]{
        \includegraphics[width=0.28\textwidth]{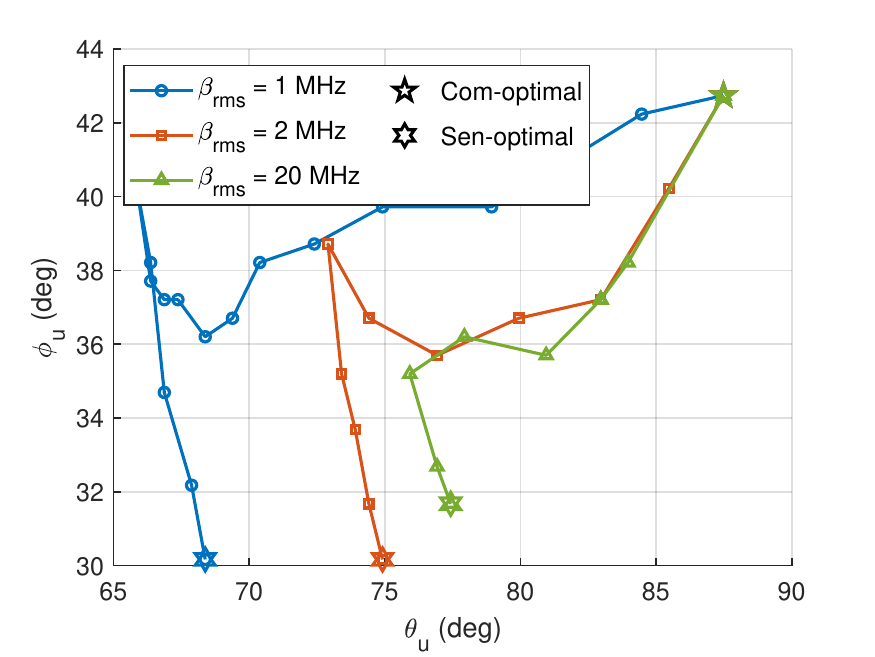}
        \label{fig:bandwidth_angles}}
    \caption{Impact of the RMS bandwidth $\beta_{\rm rms}$ on the rate-CRLB tradeoff and the optimal UAV directions for the rotated-ellipsoidal VoI. Each point in (c) corresponds to one Pareto rate sample, and the points are connected according to their order along the Pareto boundary.}
    \label{fig:bandwidth_results}
\end{figure*}

As shown in Fig.~\ref{fig:bandwidth_front}, increasing $\beta_{\rm rms}$ improves the regional localization performance over the entire Pareto boundary, while leaving the maximum achievable communication rate unchanged. This is because the communication rate is independent of the sensing waveform bandwidth, whereas the delay-information coefficient satisfies $\xi_\tau\propto\beta_{\rm rms}^{2}$, such that the delay-related CRLB term decreases approximately with $\beta_{\rm rms}^{-2}$. The improvement is particularly pronounced near the sensing-oriented end of the boundary. Nevertheless, the gain exhibits a clear diminishing-return behavior. The boundaries for $\beta_{\rm rms}=2$ MHz and $20$ MHz become nearly indistinguishable over a large portion of the rate range, whereas the boundary for $\beta_{\rm rms}=1$ MHz remains visibly higher. Once the delay contribution has been sufficiently reduced, the regional CRLB is mainly limited by the angular and delay-angle coupling information, which cannot be further improved by increasing the waveform bandwidth alone. 

The angular placements in Fig.~\ref{fig:bandwidth_angles} further show that the bandwidth affects not only the CRLB value but also the optimal UAV direction. At the sensing-optimal endpoint, increasing $\beta_{\rm rms}$ shifts the elevation angle from approximately $68^\circ$ for $1$ MHz to $75^\circ$ for $2$ MHz and about $78^\circ$ for $20$ MHz, while the corresponding azimuth angles remain within a relatively narrow range of approximately $30^\circ$-$32^\circ$. As the communication rate increases, the three angular trajectories gradually approach one another and eventually converge to nearly the same communication-optimal direction. In this high-rate regime, the minimum admissible BS-UAV distance and the associated communication requirement largely determine the placement, leaving less freedom for directional adjustment. By contrast, near the sensing-oriented endpoint, the UAV has greater spatial freedom, and changing the relative weight of the delay information leads to visibly different optimal directions.

% These results demonstrate that the waveform bandwidth can reshape both the achievable rate-CRLB boundary and the associated placement directions, but its influence is strongly regime dependent. A small RMS bandwidth makes delay estimation a dominant localization bottleneck, whereas sufficiently large bandwidth shifts the limiting factor toward the array-provided angular information. Therefore, increasing $\beta_{\rm rms}$ beyond this transition mainly yields diminishing CRLB gains and only minor changes in the Pareto placement.

\subsection{Impact of the UPA Configuration}
\label{subsec:sim_upa}

We finally investigate how the UPA configuration affects the rate-CRLB tradeoff and the corresponding optimal UAV placement. Three array configurations, namely $8\times8$, $4\times16$, and $16\times4$, are considered. All configurations contain the same total number of antennas, i.e., $N_{\rm R}=64$, and use the same half-wavelength inter-element spacing.

\begin{figure*}[t]
    \centering
    \subfloat[Rate-CRLB Pareto boundaries.]{
        \includegraphics[width=0.28\textwidth]{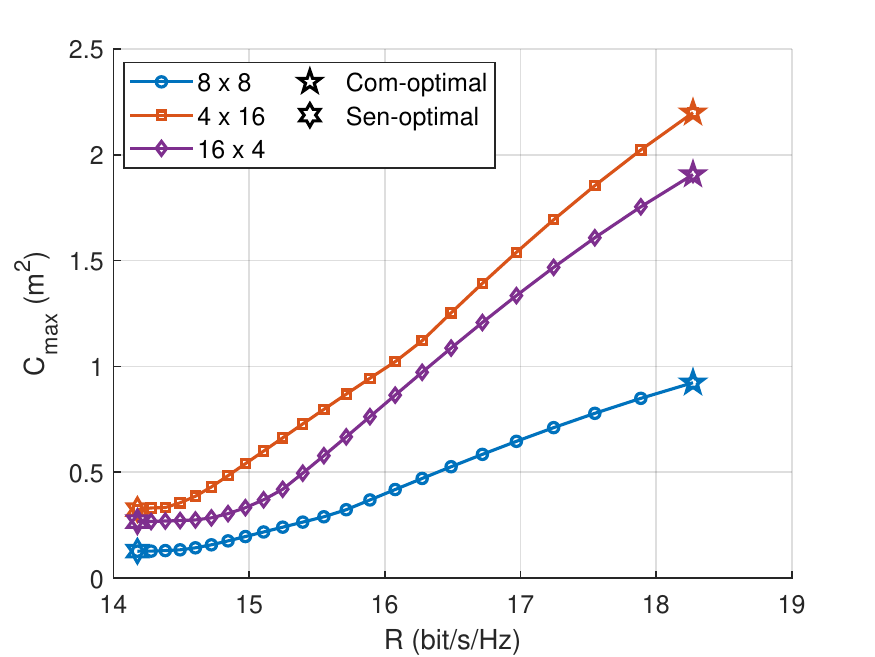}
        \label{fig:upa_front}}
    \subfloat[Optimal UAV placements.]{
        \includegraphics[width=0.28\textwidth]{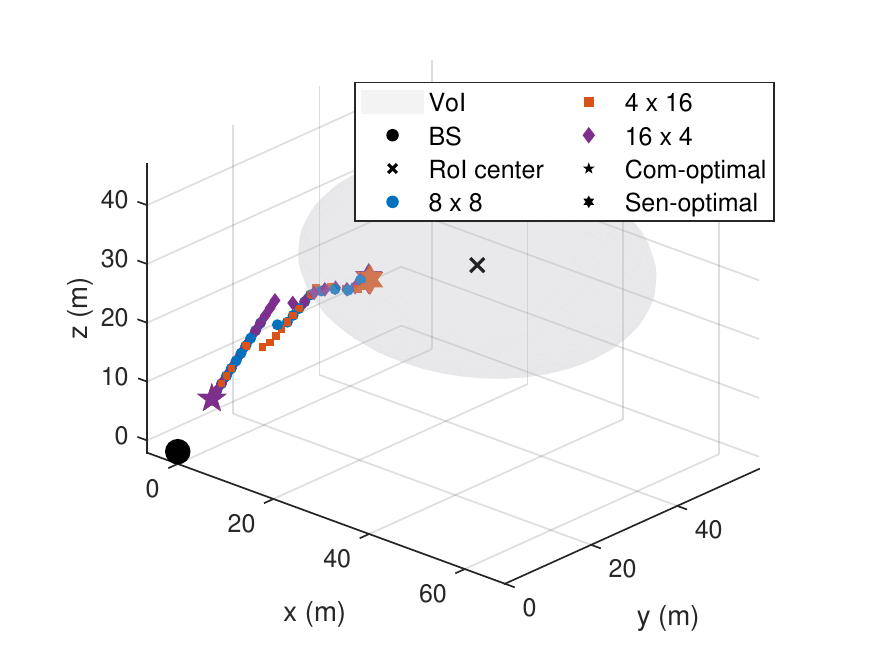}
        \label{fig:upa_placement}} 
    \subfloat[Optimal UAV directions.]{
        \includegraphics[width=0.28\textwidth]{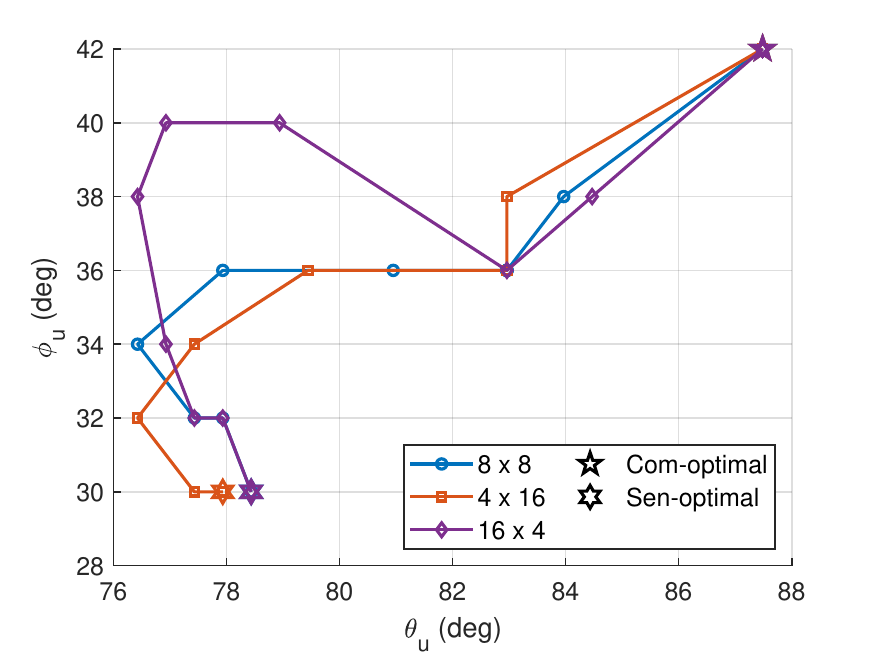}
        \label{fig:upa_angles}}
    \caption{Impact of the UPA configuration on the rate-CRLB tradeoff and the optimal UAV placement directions for the rotated-ellipsoidal VoI.}
    \label{fig:upa_results}
\end{figure*}

Fig.~\ref{fig:upa_front} shows that the square $8\times8$ UPA achieves the smallest regional worst-case CRLB throughout the considered rate range. A more substantial gap appears at the communication-optimal endpoint than that at the sensing-optimal endpoint. Although each rectangular UPA provides a larger aperture along one axis, it simultaneously reduces the number of elements and the angular information along the orthogonal axis. The balanced $8\times8$ configuration provides more uniform angular resolution over the two array dimensions and consequently yields a substantially lower worst-case CRLB. The two rectangular configurations also exhibit different performance despite having identical total antenna numbers and interchanged dimensions. In particular, the $16\times4$ UPA consistently outperforms the $4\times16$ UPA. This result shows that the two array axes are not interchangeable for the considered off-axis rotated VoI, where the target-region orientation and the corresponding bistatic geometry place a greater demand on the information provided by the $y$-axis aperture. This ordering is geometry dependent and may change when the VoI orientation or shape is varied.

Fig.~\ref{fig:upa_angles} further demonstrates that the UPA configuration changes not only the attainable CRLB but also the optimal UAV placement. In the low- and moderate-rate regimes, the three angular trajectories differ noticeably. The $8\times8$ and $4\times16$ arrays select relatively compact directional paths, whereas the $16\times4$ array exhibits a wider variation in $\phi_{\rm u}$. These changes suggest that modifying the relative angular-information weights can alter the active worst-case target geometry and cause the optimal UAV direction to switch between different directional branches as the BS-UAV distance varies. Near the communication-optimal endpoint, however, all three trajectories converge to nearly the same direction. Thus, under the stringent high-rate requirement, the placement direction becomes comparatively insensitive to the array aspect ratio in the considered scenario. Nevertheless, the corresponding CRLB values remain markedly different. 

% Overall, these results show that the UPA aspect ratio is a genuine placement-shaping parameter rather than a simple CRLB scaling factor. A balanced array can be preferable for robust regional localization because the worst-case objective penalizes weak angular resolution along either array dimension. For an anisotropic RoI, the preferred rectangular orientation is further determined by the relative alignment between the array axes and the target-region geometry.

\section{Conclusion}\label{sec:conclusion}

This paper characterized the Pareto-optimal rate-CRLB tradeoff via 3D placement of a cooperative UAV for bistatic localization over a VoI. For a singleton VoI, the complete Pareto set was proved to be the BS-target line segment, and the corresponding boundary was derived in closed form. For the radially truncated spherical sector VoI, the Pareto set was confined to the sector axis and the inner radial region under the stated conditions. For a general VoI, inner-region convexity and a Pareto-radius bound led to a two-stage computation procedure. Numerical results showed that sensing-SNR-based and VoI-center-based placements can be dominated or yield an infinite worst-case CRLB. They also showed that the RMS bandwidth and UPA configuration can change the optimal UAV direction, rather than merely rescaling the localization error.

% \section*{Acknowledgments}

%{\appendices
%\section*{Proof of the First Zonklar Equation}
%Appendix one text goes here.
% You can choose not to have a title for an appendix if you want by leaving the argument blank
%\section*{Proof of the Second Zonklar Equation}
%Appendix two text goes here.}

\appendices

\section{Proof of Theorem~\ref{the:singleton_pareto}}
\label{app:singleton_pareto}

For the singleton VoI $\mathcal A=\{\bar{\mathbf q}_{\rm t}\}$, the target-related quantities $\bar r_{\rm t}$, $\bar\ell_x$, $\bar\ell_y$, and $\bar\ell_z$ are fixed. Define
\begin{equation}
\begin{aligned}
    \Gamma_{\rm t}
    \triangleq{}&
    \frac{c^2\bar r_{\rm t}^2}{4\xi_\tau}
    +
    \frac{\bar r_{\rm t}^4}{\bar\ell_x^2}
    \left(
    \frac{1-\bar\ell_z^2}{\xi_y}
    +
    \frac{1-\bar\ell_y^2}{\xi_z}
    \right)
    >0.
\end{aligned}
\label{eq:gamma_t_appendix}
\end{equation}
For any nonsingular UAV position $\mathbf q_{\rm u}$, the bistatic factor in \eqref{eq:crlb_geometric_form} satisfies $0<h\leq 2$.
% \begin{equation}
%     0<h=1+\cos\psi\leq 2.
% \label{eq:h_upper_bound_singleton}
% \end{equation}
Since the coupling term in \eqref{eq:crlb_geometric_form} is nonnegative, it follows that
\begin{equation}
\begin{aligned}
    \mathcal C(\bar{\mathbf q}_{\rm t},\mathbf q_{\rm u})
    \geq
    \frac{c^2r_{\rm ut}^2\bar r_{\rm t}^2}{4\xi_\tau}
    +
    \frac{r_{\rm ut}^2\bar r_{\rm t}^4}{\bar\ell_x^2}
    \left(
    \frac{1-\bar\ell_z^2}{\xi_y}
    +
    \frac{1-\bar\ell_y^2}{\xi_z}
    \right) 
    =\Gamma_{\rm t}r_{\rm ut}^2.
\end{aligned}
\label{eq:singleton_crlb_lower_bound}
\end{equation}
Equality in \eqref{eq:singleton_crlb_lower_bound} is attained when the UAV lies on the BS-target line segment, since in this case $h=2$ and $S_{xy}=S_{xz}=0$.

We first show that any UAV position outside $\mathcal L_{\rm t}$ is dominated by a location on $\mathcal L_{\rm t}$. For an arbitrary $\mathbf q_{\rm u}\in\mathcal Q$, define
\begin{equation}
    d\triangleq\min\{r_{\rm ut},\bar r_{\rm t}\},
    \qquad
    \widetilde{\mathbf q}_{\rm u}
    \triangleq
    \left(1-\frac{d}{\bar r_{\rm t}}\right)
    \bar{\mathbf q}_{\rm t}.
\label{eq:projected_singleton_point}
\end{equation}
By $\mathcal L_{\rm t}\subseteq\mathcal Q$, the location $\widetilde{\mathbf q}_{\rm u}$ is feasible. Moreover,
\begin{equation}
    \widetilde r_{\rm ut}=d,
    \qquad
    \widetilde r_{\rm u}=\bar r_{\rm t}-d,
    \qquad
    \mathcal C(\bar{\mathbf q}_{\rm t},\widetilde{\mathbf q}_{\rm u})
    =
    \Gamma_{\rm t}d^2.
\label{eq:projected_singleton_performance}
\end{equation}
If $r_{\rm ut}<\bar r_{\rm t}$, the reverse triangle inequality gives $r_{\rm u}\geq\bar r_{\rm t}-r_{\rm ut}=\widetilde r_{\rm u}$.
% \begin{equation}
%     r_{\rm u}
%     \geq
%     \bar r_{\rm t}-r_{\rm ut}
%     =
%     \widetilde r_{\rm u}.
% \label{eq:singleton_rate_domination_case1}
% \end{equation}
If $r_{\rm ut}\geq\bar r_{\rm t}$, then $\widetilde r_{\rm u}=0\leq r_{\rm u}$. Hence, in both cases, $\widetilde r_{\rm u}\leq r_{\rm u}$.
% \begin{equation}
%     \widetilde r_{\rm u}\leq r_{\rm u}.
% \label{eq:singleton_rate_domination}
% \end{equation}
Furthermore, since $d\leq r_{\rm ut}$, \eqref{eq:singleton_crlb_lower_bound} yields
\begin{equation}
\begin{aligned}
    \mathcal C(\bar{\mathbf q}_{\rm t},\widetilde{\mathbf q}_{\rm u})
    =
    \Gamma_{\rm t}d^2
    \leq
    \Gamma_{\rm t}r_{\rm ut}^2
    \leq
    \mathcal C(\bar{\mathbf q}_{\rm t},\mathbf q_{\rm u}).
\end{aligned}
\label{eq:singleton_crlb_domination}
\end{equation}
For any $\mathbf q_{\rm u}\notin\mathcal L_{\rm t}$, at least one of $\widetilde r_{\rm u}\leq r_{\rm u}$ and \eqref{eq:singleton_crlb_domination} is strict. Therefore, every UAV position outside $\mathcal L_{\rm t}$ is dominated by a feasible UAV position on $\mathcal L_{\rm t}$.

It remains to show that no location on $\mathcal L_{\rm t}$ is dominated. Parameterize a location on this segment by its UAV-target distance $d$, i.e.,
\begin{equation}
    \mathbf q_{\rm u}(d)
    =
    \left(1-\frac{d}{\bar r_{\rm t}}\right)
    \bar{\mathbf q}_{\rm t},
    \qquad
    0\leq d\leq\bar r_{\rm t}.
\label{eq:line_parameterization_singleton}
\end{equation}
Its sensing and communication objectives satisfy
\begin{equation}
    \mathcal C(\bar{\mathbf q}_{\rm t},\mathbf q_{\rm u}(d))
    =
    \Gamma_{\rm t}d^2,
    \qquad
    r_{\rm u}(d)
    =
    \bar r_{\rm t}-d.
\label{eq:line_objectives_singleton}
\end{equation}
Thus, increasing $d$ strictly increases the CRLB but strictly decreases the BS-UAV distance. Hence, no two distinct points on $\mathcal L_{\rm t}$ dominate each other.
Consequently, the Pareto-optimal UAV placement set is
\begin{equation}
    \mathcal L_{\rm Pareto}^{\rm sing}
    =
    \left\{
    \mathbf q_{\rm u}^{\star}
    =
    \alpha\bar{\mathbf q}_{\rm t}
    \mid
    0\leq\alpha\leq 1
    \right\},
\end{equation}
where the endpoint $\alpha=1$ is understood in the limiting sense under the ideal point-target propagation model. Equivalently,
\begin{equation}
    r_{\rm u}^{\star}
    =
    \bar r_{\rm t}-r_{\rm ut}^{\star},
    \qquad
    0\leq r_{\rm ut}^{\star}\leq\bar r_{\rm t}.
\end{equation}
This completes the proof.

\section{Proof of Proposition~\ref{prop:crlb_inner_convexity}}
\label{app:crlb_inner_convexity}

For a fixed target, define $\mathbf e_{\rm t}\triangleq\mathbf q_{\rm t}/r_{\rm t}$, $\mathbf p\triangleq\mathbf q_{\rm t}-\mathbf q_{\rm u}$, $d\triangleq\|\mathbf p\|$ and $u\triangleq\mathbf e_{\rm t}^{\rm T}\mathbf p$.
% \begin{equation}
%     \mathbf e_{\rm t}
%     \triangleq
%     \mathbf q_{\rm t}/r_{\rm t},
%     \quad
%     \mathbf p
%     \triangleq
%     \mathbf q_{\rm t}-\mathbf q_{\rm u},
%     \quad
%     d\triangleq\|\mathbf p\|,
%     \quad
%     u\triangleq\mathbf e_{\rm t}^{\rm T}\mathbf p.
%     \label{eq:inner_convexity_variables}
% \end{equation}
If $\|\mathbf q_{\rm u}\|<r_{\rm t}$, then
\begin{equation}
    u
    =
    r_{\rm t}
    -
    \mathbf e_{\rm t}^{\rm T}\mathbf q_{\rm u}
    \geq
    r_{\rm t}-\|\mathbf q_{\rm u}\|
    >0,
    \quad
    h=(d+u)/d.
    \label{eq:inner_convexity_geometry}
\end{equation}
Define the unit vectors
\begin{equation}
    \mathbf m_y
    \triangleq
    \frac{\mathbf e_z\times\mathbf e_{\rm t}}
    {\sqrt{1-\ell_z^2}},
    \qquad
    \mathbf m_z
    \triangleq
    \frac{\mathbf e_y\times\mathbf e_{\rm t}}
    {\sqrt{1-\ell_y^2}},
    \label{eq:inner_convexity_vectors}
\end{equation}
where a term is omitted when its coefficient is zero. Since $\mathbf q_{\rm t}\times\mathbf q_{\rm u}=-r_{\rm t}\mathbf e_{\rm t}\times\mathbf p$, the projected areas satisfy
\begin{equation}
    4S_{xy}^{2}
    =
    r_{\rm t}^{2}(1-\ell_z^{2})
    (\mathbf m_y^{\rm T}\mathbf p)^2,
    \qquad
    4S_{xz}^{2}
    =
    r_{\rm t}^{2}(1-\ell_y^{2})
    (\mathbf m_z^{\rm T}\mathbf p)^2.
    \label{eq:inner_convexity_projected_areas}
\end{equation}
Substituting \eqref{eq:inner_convexity_geometry} and \eqref{eq:inner_convexity_projected_areas} into \eqref{eq:crlb_geometric_form} gives
\begin{equation}
    \mathcal C(\mathbf q_{\rm t},\mathbf q_{\rm u})
    =
    a_{\tau}\phi_{\tau}(\mathbf p)
    +
    a_y\phi_{\mathbf m_y}(\mathbf p)
    +
    a_z\phi_{\mathbf m_z}(\mathbf p),
    \label{eq:inner_convexity_decomposition}
\end{equation}
where
\begin{equation}
\begin{aligned}
    \phi_{\tau}(\mathbf p)
    \triangleq
    \frac{d^{4}}{(d+u)^2}, \quad
    \phi_{\mathbf m}(\mathbf p)
    \triangleq
    d^2+
    \frac{(\mathbf m^{\rm T}\mathbf p)^2d^2}
    {(d+u)^2},\\
    a_{\tau}
    \triangleq
    \frac{c^2r_{\rm t}^{2}}{\xi_{\tau}},\quad
    a_y
    \triangleq
    \frac{r_{\rm t}^{4}(1-\ell_z^2)}
    {\ell_x^2\xi_y},\quad
    a_z
    \triangleq
    \frac{r_{\rm t}^{4}(1-\ell_y^2)}
    {\ell_x^2\xi_z}.
\end{aligned}
\label{eq:inner_convexity_components}
\end{equation}

% \begin{equation}
% \begin{aligned}
%     \phi_{\tau}(\mathbf p)
%     &\triangleq
%     \frac{d^{4}}{(d+u)^2},\\
%     \phi_{\mathbf m}(\mathbf p)
%     &\triangleq
%     d^2+
%     \frac{(\mathbf m^{\rm T}\mathbf p)^2d^2}
%     {(d+u)^2},
% \end{aligned}
% \qquad
% \begin{aligned}
%     a_{\tau}
%     &\triangleq
%     \frac{c^2r_{\rm t}^{2}}{\xi_{\tau}},\\
%     a_y
%     &\triangleq
%     \frac{r_{\rm t}^{4}(1-\ell_z^2)}
%     {\ell_x^2\xi_y},\\
%     a_z
%     &\triangleq
%     \frac{r_{\rm t}^{4}(1-\ell_y^2)}
%     {\ell_x^2\xi_z}.
% \end{aligned}
% \label{eq:inner_convexity_components}
% \end{equation}
% For any unit vectors $\mathbf e_{\rm t}$ and $\mathbf m$ satisfying $\mathbf m^{\rm T}\mathbf e_{\rm t}=0$, direct differentiation over the half-space $u=\mathbf e_{\rm t}^{\rm T}\mathbf p>0$ gives
% \begin{equation}
%     \nabla_{\mathbf p}^{2}\phi_{\tau}(\mathbf p)
%     \succeq\mathbf 0,
%     \qquad
%     \nabla_{\mathbf p}^{2}\phi_{\mathbf m}(\mathbf p)
%     \succeq\mathbf 0.
%     \label{eq:inner_component_hessians}
% \end{equation}

It remains to establish the convexity of $\phi_\tau$ and $\phi_{\mathbf m}$. Define $\eta\triangleq u/d\in(0,1]$. For $\phi_\tau$, choose a unit vector $\mathbf b$ such that $\mathbf e_{\rm t}=\eta\mathbf p/d+\sqrt{1-\eta^2}\mathbf b$, and decompose an arbitrary direction as $\mathbf v=a\mathbf p/d+b\mathbf b+\mathbf v_{\perp}$, where $\mathbf v_{\perp}$ is orthogonal to both $\mathbf p$ and $\mathbf b$. Direct differentiation gives
\begin{equation}
\begin{aligned}
    \mathbf v^{\rm T}
    \nabla_{\mathbf p}^2\phi_\tau(\mathbf p)
    \mathbf v
    =
    \frac{2}{(1+\eta)^4}
    \bigg[
    \begin{bmatrix}a&b\end{bmatrix}
    \mathbf M_\tau(\eta)
    \begin{bmatrix}a\\b\end{bmatrix}\\
    +
    (1+\eta)(2\eta+1)\|\mathbf v_{\perp}\|^2
    \bigg],
\end{aligned}
\label{eq:app_phi_tau_second_derivative}
\end{equation}
where
\begin{equation}
    \mathbf M_\tau(\eta)
    \triangleq
    \begin{bmatrix}
    (1+\eta)^2 &
    -(1+\eta)\sqrt{1-\eta^2}\\
    -(1+\eta)\sqrt{1-\eta^2} &
    (1+\eta)(4-\eta)
    \end{bmatrix}.
    \label{eq:app_phi_tau_matrix}
\end{equation}
Since the diagonal entries of $\mathbf M_\tau(\eta)$ are positive and $\det\mathbf M_\tau(\eta)=3(1+\eta)^3>0$, we have $\mathbf M_\tau(\eta)\succ\mathbf 0$ and hence $\mathbf v^{\rm T}\nabla_{\mathbf p}^2\phi_\tau(\mathbf p)\mathbf v \geq 0$.
For $\phi_{\mathbf m}$, complete $\{\mathbf e_{\rm t},\mathbf m\}$ to an orthonormal basis $\{\mathbf e_{\rm t},\mathbf m,\mathbf n\}$ and write
\begin{equation}
    \mathbf p/d
    =
    \eta\mathbf e_{\rm t}
    +
    \nu\mathbf m
    +
    \omega\mathbf n,
    \qquad
    \eta^2+\nu^2+\omega^2=1.
    \label{eq:app_phi_m_coordinates}
\end{equation}
Since $\phi_{\mathbf m}$ is homogeneous of degree two, its Hessian depends only on $\eta$, $\nu$, and $\omega$. In the above basis, its leading principal minors are
\begin{equation}
\begin{aligned}
    \Delta_1
    &=
    \frac{2\left[(1+\eta)^2+3(1-\eta)\nu^2\right]}
    {(1+\eta)^2},\\
    \Delta_2
    &=
    \frac{4Q(\eta,s)}{(1+\eta)^6},\qquad
    \Delta_3
    =
    \frac{8P(\eta,\nu^2)}{(1+\eta)^7},
\end{aligned}
\label{eq:app_phi_m_principal_minors}
\end{equation}
where $s\triangleq\nu^2/(1-\eta^2)\in[0,1]$ for $\eta<1$, with the case $\eta=1$ obtained by continuity, and
\begin{equation}
\begin{aligned}
    Q(\eta,s)
    &=
    B_0(1-s)^3
    +3B_1s(1-s)^2
    +3B_2s^2(1-s)
    +B_3s^3,\\
    B_0
    &=
    (1+\eta)^4(\eta^2+2\eta+2),\\
    B_1
    &=
    \frac{(1+\eta)^3}{3}
    \left(
    3\eta^4+2\eta^3-6\eta^2+23\eta+8
    \right),\\
    B_2
    &=
    \frac{(1+\eta)^3}{3}
    \left(
    9\eta^3-15\eta^2+22\eta+14
    \right),\\
    B_3
    &=
    (1+\eta)^3(3\eta+7),
\end{aligned}
\label{eq:app_phi_m_Q}
\end{equation}
and
\begin{equation}
\begin{aligned}
    P(\eta,z)
    ={}&
    \eta^7+7\eta^6+22\eta^5+40\eta^4+45\eta^3+31\eta^2+12\eta+2\\
    &+
    z\left(
    3\eta^5+15\eta^4+32\eta^3+33\eta^2+15\eta+2
    \right)\\
    &+
    z^2\left(
    3\eta^3+9\eta^2+10\eta+2
    \right)
    +
    z^3(1+\eta).
\end{aligned}
\label{eq:app_phi_m_P}
\end{equation}
For $0<\eta\leq1$, all $B_i$ are positive and $P(\eta,z)>0$ for $z\geq0$. Hence, $\Delta_1$, $\Delta_2$, and $\Delta_3$ are positive, and Sylvester's criterion gives $\nabla_{\mathbf p}^2\phi_{\mathbf m}(\mathbf p)\succ\mathbf 0$. Therefore, both $\phi_\tau$ and $\phi_{\mathbf m}$ are convex over $u>0$. Since $a_\tau$, $a_y$, and $a_z$ are nonnegative, \eqref{eq:inner_convexity_decomposition} is convex in $\mathbf p$ and hence in $\mathbf q_{\rm u}$ because $\mathbf p=\mathbf q_{\rm t}-\mathbf q_{\rm u}$ is affine. Finally, the pointwise maximum of convex functions is convex, which proves the convexity of $C_{\max}(\mathbf q_{\rm u})$ over $\mathcal{Q}_1$.

\section{Proof of Lemma~\ref{lem:Q1_axis}}
\label{app:Q1_axis}
Consider an arbitrary $\mathbf q_{\rm u}=[x_{\rm u},y_{\rm u},z_{\rm u}]^{\rm T}\in\mathcal Q_1$. The VoI is invariant under reflections about the $xy$- and $xz$-planes. Moreover, the CRLB is invariant when the target and UAV locations are simultaneously reflected about either plane, even when $N_y\neq N_z$. Hence, the four UAV locations $\mathbf q_{\rm u}^{(i)}=[x_{\rm u},\pm y_{\rm u},\pm z_{\rm u}]^{\rm T}$ have the same regional worst-case CRLB. Their average is $x_{\rm u}\mathbf e_x$. By Proposition~\ref{prop:crlb_inner_convexity},
\begin{equation}
    C_{\max}(x_{\rm u}\mathbf e_x)
    \leq
    \frac{1}{4}\sum_{i=1}^{4}
    C_{\max}(\mathbf q_{\rm u}^{(i)})
    =
    C_{\max}(\mathbf q_{\rm u}).
\end{equation}
If $(y_{\rm u},z_{\rm u})\neq(0,0)$, then $\|x_{\rm u}\mathbf e_x\|<\|\mathbf q_{\rm u}\|$, so $x_{\rm u}\mathbf e_x$ also achieves a strictly higher communication rate. Therefore, every off-axis location in $\mathcal Q_1$ is dominated, and any Pareto-optimal location in $\mathcal Q_1$ must lie on the $x$-axis.

\section{Proof of Lemma~\ref{lem:Q2_dominated}}
\label{app:Q2_dominated}

Consider an arbitrary UAV location $\mathbf q_{\rm u}^{(2)}\in\mathcal Q_2$, and denote its distance and off-axis angle by $r_{\rm u}^{(2)}$ and $\vartheta_{\rm u}^{(2)}$. If $\vartheta_{\rm u}^{(2)} \le \vartheta_{\max}$, the $C_{\max}(\mathbf{q}_{\rm u})$ becomes unbounded and it is obviously dominated by some location in $\mathcal{Q}_1$. Next we consider when $\vartheta_{\rm u}^{(2)} > \vartheta_{\max}$.
% \begin{equation}
%     r_{\rm u}^{(2)}\geq r_{\min},
%     \qquad
%     \vartheta_{\rm u}^{(2)}>\vartheta_{\max}.
%     \label{eq:Q2_arbitrary_uav}
% \end{equation}
By the rotational symmetry of the radially truncated spherical sector VoI and the square-UPA CRLB, we may represent $\mathbf q_{\rm u}^{(2)}$ in an arbitrary axial plane without loss of generality.

First, we construct a feasible on-axis UAV location in $\mathcal Q_1$. From
$\vartheta_{\rm u}^{(2)}>\vartheta_{\max}$ and
\eqref{eq:Q2_domination_condition}, we have
\begin{equation}
\begin{aligned}
    r_{\max}\cos\left(
    \vartheta_{\rm u}^{(2)}+\vartheta_{\max}
    \right)
    <
    r_{\max}\cos(2\vartheta_{\max}) 
    \leq r_{\min}.
\end{aligned}
\label{eq:radial_comparator_existence}
\end{equation}
Therefore, there exists $r_0$ satisfying
\begin{equation}
    \max\left\{
    0,\,
    r_{\max}\cos\left(
    \vartheta_{\rm u}^{(2)}+\vartheta_{\max}
    \right)
    \right\}
    <r_0<r_{\min}.
    \label{eq:r0_selection}
\end{equation}
Define the on-axis UAV location as $\mathbf q_{\rm u}^{(0)}\triangleq r_0\mathbf e_x$.
% \begin{equation}
%     \mathbf q_{\rm u}^{(0)}
%     \triangleq
%     r_0\mathbf e_x.
%     \label{eq:axis_dominating_uav}
% \end{equation}
Clearly, $\mathbf q_{\rm u}^{(0)}\in\mathcal Q_1$ and
$r_0<r_{\min}\leq r_{\rm u}^{(2)}$. Since the communication rate is decreasing with the BS-UAV distance, we have
\begin{equation}
    R(\mathbf q_{\rm u}^{(0)})
    >
    R(\mathbf q_{\rm u}^{(2)}).
    \label{eq:Q2_rate_domination}
\end{equation}

It remains to compare their regional worst-case CRLBs. We first show that, for the on-axis location $\mathbf q_{\rm u}^{(0)}$, a worst-case target can be chosen on the boundary of the cone. Consider a target with distance $r_{\rm t}$ and off-axis angle $\vartheta_{\rm t}$. Owing to the axial symmetry of $\mathbf q_{\rm u}^{(0)}$, its CRLB is independent of the target circumferential angle. Let $\delta=\vartheta_{\rm t}$ denote the angular separation between the BS-UAV and BS-target directions, and define
\begin{equation}
    d
    \triangleq
    \left(
    r_0^2+r_{\rm t}^2-2r_0r_{\rm t}\cos\delta
    \right)^{1/2},
    \qquad
    h
    \triangleq
    1+
    \frac{r_{\rm t}-r_0\cos\delta}{d}.
    \label{eq:axis_target_geometry}
\end{equation}
Since $r_0<r_{\min}\leq r_{\rm t}$, the CRLB in
\eqref{eq:crlb_conical_geometric_form} can be written as
\begin{equation}
\begin{aligned}
    \mathcal C_{\delta}
    ={}&
    \frac{c^2r_{\rm t}^2}{\xi_\tau}
    \left(\frac{d}{h}\right)^2
    +
    \frac{r_{\rm t}^4}{\xi_y}
    d^2 
    +
    \frac{2r_{\rm t}^4}{\xi_y\cos^2\delta}
    \frac{d^2}{h},
\end{aligned}
\label{eq:axis_crlb_delta}
\end{equation}
where the identity $\frac{r_0^2\sin^2\delta}{h^2}=d^2\left(\frac{2}{h}-1\right)$
% \begin{equation}
%     \frac{r_0^2\sin^2\delta}{h^2}
%     =
%     d^2\left(\frac{2}{h}-1\right)
%     \label{eq:axis_coupling_identity}
% \end{equation}
has been used. Furthermore,
\begin{equation}
    \frac{\partial d}{\partial\delta}
    =
    \frac{r_0r_{\rm t}\sin\delta}{d}
    >0,
\end{equation}
and
\begin{equation}
\begin{aligned}
    \frac{\partial}{\partial\delta}
    \left(\frac{d}{h}\right)
    =
    \frac{
    r_0\sin\delta
    \left[
    r_{\rm t}(d+r_{\rm t})-r_0^2
    \right]
    }
    {
    \left(
    d+r_{\rm t}-r_0\cos\delta
    \right)^2
    }
    >0,
\end{aligned}
\label{eq:axis_d_over_h_monotonic}
\end{equation}
for $0<\delta<\pi$, where the positivity follows from
$r_0<r_{\rm t}$. Hence, every term in
\eqref{eq:axis_crlb_delta} is increasing with respect to
$\delta=\vartheta_{\rm t}$. Thus, a worst-case target for
$\mathbf q_{\rm u}^{(0)}$ can be chosen with
\begin{equation}
    \vartheta_{\rm t}^{\star}
    =
    \vartheta_{\max}.
    \label{eq:axis_worst_target_angle}
\end{equation}
Let $r_{\rm t}^{\star}\in[r_{\min},r_{\max}]$ denote one corresponding worst-case radial coordinate. Since
$\mathbf q_{\rm u}^{(0)}$ lies on the cone axis, the circumferential direction of this worst-case target can be arbitrarily selected. We choose it on the axial plane opposite to
$\mathbf q_{\rm u}^{(2)}$, and denote the resulting target position by $\mathbf q_{\rm t}^{-}$. Then
\begin{equation}
    \mathcal C_{\max}(\mathbf q_{\rm u}^{(0)})
    =
    \mathcal C(\mathbf q_{\rm t}^{-},\mathbf q_{\rm u}^{(0)}).
    \label{eq:axis_worst_target_selected}
\end{equation}

For the selected target $\mathbf q_{\rm t}^{-}$, the angular separations corresponding to $\mathbf q_{\rm u}^{(0)}$ and
$\mathbf q_{\rm u}^{(2)}$ are respectively
\begin{equation}
    \delta_0=\vartheta_{\max},
    \qquad
    \delta_2=
    \vartheta_{\rm u}^{(2)}+\vartheta_{\max}
    >
    \delta_0.
    \label{eq:separation_comparison}
\end{equation}
Consider first a virtual UAV location having distance $r_0$ and the same direction as $\mathbf q_{\rm u}^{(2)}$. For fixed $r_0$, $r_{\rm t}^{\star}$, and $\vartheta_{\max}$, the same monotonicity argument as in
\eqref{eq:axis_d_over_h_monotonic} shows that the CRLB strictly increases with the separation angle $\delta$. Hence,
\begin{equation}
    \mathcal C(\mathbf q_{\rm t}^{-},
    r_0\mathbf e_{\rm u}^{(2)})
    >
    \mathcal C(\mathbf q_{\rm t}^{-},
    \mathbf q_{\rm u}^{(0)}),
    \label{eq:angular_domination_step}
\end{equation}
where
$\mathbf e_{\rm u}^{(2)}
=\mathbf q_{\rm u}^{(2)}/r_{\rm u}^{(2)}$.

We next increase the UAV distance along the direction
$\mathbf e_{\rm u}^{(2)}$ from $r_0$ to $r_{\rm u}^{(2)}$.
For the fixed target $\mathbf q_{\rm t}^{-}$, let $s$ denote the UAV distance on this ray and let
$\delta=\delta_2$. The corresponding UAV-target distance and bistatic factor are
\begin{equation}
\begin{aligned}
    &d(s)
    =
    \left(
    s^2+(r_{\rm t}^{\star})^2
    -
    2sr_{\rm t}^{\star}\cos\delta_2
    \right)^{1/2},\\
    &h(s)
    =
    1+
    \frac{
    r_{\rm t}^{\star}-s\cos\delta_2
    }{d(s)}.
    \label{eq:radial_comparison_geometry}
\end{aligned}
\end{equation}
By \eqref{eq:r0_selection} and
$r_{\rm t}^{\star}\leq r_{\max}$, we have $s\geq r_0>r_{\rm t}^{\star}\cos\delta_2$.
% \begin{equation}
%     s\geq r_0
%     >
%     r_{\rm t}^{\star}\cos\delta_2.
%     \label{eq:radial_monotonic_condition}
% \end{equation}
Therefore,
\begin{equation}
    \frac{{\rm d}d(s)}{{\rm d}s}
    =
    \frac{
    s-r_{\rm t}^{\star}\cos\delta_2
    }{d(s)}
    >0,
    \qquad
    \frac{{\rm d}h(s)}{{\rm d}s}
    =
    -
    \frac{
    sr_{\rm t}^{\star}\sin^2\delta_2
    }{
    d^3(s)
    }
    <0.
    \label{eq:radial_monotonic_derivatives}
\end{equation}
For fixed $\mathbf q_{\rm t}^{-}$ and fixed UAV direction, the terms in
\eqref{eq:crlb_conical_geometric_form} are respectively proportional to $\frac{d^2(s)}{h^2(s)}$, $d^2(s)$, and $\frac{s^2}{h^2(s)}$
% \begin{equation}
%     \frac{d^2(s)}{h^2(s)},
%     \qquad
%     d^2(s),
%     \qquad
%     \frac{s^2}{h^2(s)},
% \end{equation}
with positive coefficients. It follows from
\eqref{eq:radial_monotonic_derivatives} that the CRLB strictly increases with $s$ over
$[r_0,r_{\rm u}^{(2)}]$. Therefore,
\begin{equation}
    \mathcal C(\mathbf q_{\rm t}^{-},
    \mathbf q_{\rm u}^{(2)})
    >
    \mathcal C(\mathbf q_{\rm t}^{-},
    r_0\mathbf e_{\rm u}^{(2)}).
    \label{eq:radial_domination_step}
\end{equation}
Combining \eqref{eq:axis_worst_target_selected},
\eqref{eq:angular_domination_step}, and
\eqref{eq:radial_domination_step}, we obtain
\begin{equation}
\begin{aligned}
    \mathcal C_{\max}(\mathbf q_{\rm u}^{(2)})
    &\geq
    \mathcal C(\mathbf q_{\rm t}^{-},
    \mathbf q_{\rm u}^{(2)}) 
    >
    \mathcal C(\mathbf q_{\rm t}^{-},
    \mathbf q_{\rm u}^{(0)})
    =
    \mathcal C_{\max}(\mathbf q_{\rm u}^{(0)}).
\end{aligned}
\label{eq:Q2_crlb_domination}
\end{equation}
Together with \eqref{eq:Q2_rate_domination}, this shows that
$\mathbf q_{\rm u}^{(0)}\in\mathcal Q_1$ strictly dominates
$\mathbf q_{\rm u}^{(2)}\in\mathcal Q_2$. Since
$\mathbf q_{\rm u}^{(2)}$ was arbitrarily selected, every location in $\mathcal Q_2$ is dominated by a location in $\mathcal Q_1$. This completes the proof.

\section{Proof of Theorem~\ref{the:conical_RoI_pareto}}
\label{app:conical_RoI_pareto}

Under condition \eqref{eq:cone_pareto_subset_condition}, Lemma~\ref{lem:Q2_dominated} shows that every UAV location in $\mathcal Q_2$ is dominated by a feasible on-axis location in $\mathcal Q_1$. Moreover, Lemma~\ref{lem:Q1_axis} shows that every off-axis location in $\mathcal Q_1$ is dominated by a feasible on-axis location in $\mathcal Q_1$. Therefore, no Pareto-optimal UAV location can belong to $\mathcal Q_2$ or to the off-axis portion of $\mathcal Q_1$, and every Pareto-optimal UAV location must satisfy \eqref{eq:cone_pareto_axis_subset}. This completes the proof.

\section{Proof of Proposition~\ref{prop:crlb_radius_bound}}
\label{app:crlb_radius_bound}

Let $\mathbf q_{\rm u}=R\mathbf e$ and $\tilde{\mathbf q}_{\rm u}=\tilde r\mathbf e$, where $r_{\mathcal A}^{\max}\leq\tilde r<R$. For an arbitrary target $\mathbf q_{\rm t}=r_{\rm t}\mathbf e_{\rm t}\in\mathcal A$, define $\cos\gamma=\mathbf e^{\rm T}\mathbf e_{\rm t}$ and
\begin{equation}
    d(s)
    \triangleq
    \left(
    s^2+r_{\rm t}^2-2sr_{\rm t}\cos\gamma
    \right)^{1/2},
    \qquad
    h(s)
    \triangleq
    1+
    \frac{r_{\rm t}-s\cos\gamma}{d(s)}.
    \label{eq:radius_bound_geometry}
\end{equation}
Since $s\geq\tilde r\geq r_{\mathcal A}^{\max}\geq r_{\rm t}$,
\begin{equation}
    \frac{{\rm d}d(s)}{{\rm d}s}
    =
    \frac{s-r_{\rm t}\cos\gamma}{d(s)}
    \geq0,
    \qquad
    \frac{{\rm d}h(s)}{{\rm d}s}
    =
    -\frac{sr_{\rm t}\sin^2\gamma}{d^3(s)}
    \leq0.
    \label{eq:radius_bound_monotonicity}
\end{equation}
Moreover, $S_{xy}(s)$ and $S_{xz}(s)$ are proportional to $s$. Therefore, decreasing the UAV distance from $R$ to $\tilde r$ does not increase $r_{\rm ut}$ or the projected-area terms and does not decrease $h$. Since the remaining target-dependent quantities are fixed, every term in \eqref{eq:crlb_geometric_form} is nonincreasing, and hence
\begin{equation}
    \mathcal C(\mathbf q_{\rm t},\tilde r\mathbf e)
    \leq
    \mathcal C(\mathbf q_{\rm t},R\mathbf e),
    \qquad
    \forall\mathbf q_{\rm t}\in\mathcal A.
    \label{eq:radius_bound_pointwise}
\end{equation}
Taking the maximum over $\mathcal A$ gives $C_{\max}(\tilde r\mathbf e)\leq C_{\max}(R\mathbf e)$. Since $\tilde r<R$, the location $\tilde r\mathbf e$ also achieves a strictly higher communication rate and therefore dominates $R\mathbf e$. This completes the proof.

\bibliographystyle{IEEEtran}
\bibliography{IEEEabrv,0reference}

@misc{xuxiaoli,
      title={Hybrid Mono- and Bi-static {OFDM-ISAC} via {BS-UE} Cooperation: Closed-Form {CRLB} and Coverage Analysis}, 
      author={Xiaoli Xu and Yong Zeng},
      year={2026},
      eprint={2601.09057},
      archivePrefix={arXiv},
      primaryClass={cs.IT},
      url={https://arxiv.org/abs/2601.09057}, 
}

@misc{CRLB,
      title={Low-Altitude {UAV}-Assisted Bistatic {ISAC}: Closed-form {3D} {CRLB} and Coverage Analysis}, 
      author={Haoyu Jiang and Hengyou Kong and Xiaoli Xu and Yong Zeng},
      year={2026},
      eprint={2607.18766},
      archivePrefix={arXiv},
      primaryClass={eess.SP},
      url={https://arxiv.org/abs/2607.18766}, 
}

@ARTICLE{huizhi,
  author={Wang, Huizhi and Xiao, Zhiqiang and Zeng, Yong},
  journal=IEEE_Transactions_on_Signal_Processing, 
  title={Cramér-Rao Bounds for Near-Field Sensing With Extremely Large-Scale {MIMO}}, 
  year={2024},
  volume={72},
  number={},
  pages={701-717},
  doi={10.1109/TSP.2024.3350329}}

@ARTICLE{zihan,
  author={Xu, Zihan and Zhou, Zhiwen and Wu, Di and Xu, Xiaoli and Zeng, Yong},
  journal=IEEE_Internet_of_Things_Journal, 
  title={{CKM}-Enabled Joint Spatial--Doppler Domain Clutter Suppression for Low-Altitude {UAV} {ISAC}}, 
  year={2026},
  volume={13},
  number={14},
  pages={31752-31767},
  doi={10.1109/JIOT.2026.3689485}}

@misc{zeng_capacity,
      title={Capacity Characterization and Formation Optimization for Multi-User {MIMO} Communications with {UAV} Swarm}, 
      author={Yong Zeng},
      year={2026},
      eprint={2605.14298},
      archivePrefix={arXiv},
      primaryClass={cs.IT},
      url={https://arxiv.org/abs/2605.14298}, 
}

@ARTICLE{zeng_UAVswarm_twc,
  author={Zeng, Yong and Xu, Xiaoli and Jin, Shi and Zhang, Rui},
  journal=IEEE_Transactions_on_Wireless_Communications, 
  title={Simultaneous Navigation and Radio Mapping for Cellular-Connected {UAV} With Deep Reinforcement Learning}, 
  year={2021},
  volume={20},
  number={7},
  pages={4205-4220},
  doi={10.1109/TWC.2021.3056573}}

@ARTICLE{zhang_ISAC_jstsp,
  author={Zhang, J. Andrew and Liu, Fan and Masouros, Christos and Heath, Robert W. and Feng, Zhiyong and Zheng, Le and Petropulu, Athina},
  journal=IEEE_Journal_of_Selected_Topics_in_Signal_Processing, 
  title={An Overview of Signal Processing Techniques for Joint Communication and Radar Sensing}, 
  year={2021},
  volume={15},
  number={6},
  pages={1295-1315},
  doi={10.1109/JSTSP.2021.3113120}}

@ARTICLE{cui_ISAC_ieeenetwork,
  author={Cui, Yuanhao and Liu, Fan and Jing, Xiaojun and Mu, Junsheng},
  journal=IEEE_Network, 
  title={Integrating Sensing and Communications for Ubiquitous {IoT}: Applications, Trends, and Challenges}, 
  year={2021},
  volume={35},
  number={5},
  pages={158-167},
  doi={10.1109/MNET.010.2100152}}

@ARTICLE{zhang_ISAC_survey_tutorial,
  author={Zhang, J. Andrew and Rahman, Md. Lushanur and Wu, Kai and Huang, Xiaojing and Guo, Y. Jay and Chen, Shanzhi and Yuan, Jinhong},
  journal=IEEE_Communications_Surveys_Tutorials, 
  title={Enabling Joint Communication and Radar Sensing in Mobile Networks—A Survey}, 
  year={2022},
  volume={24},
  number={1},
  pages={306-345},
  doi={10.1109/COMST.2021.3122519}}

@ARTICLE{wei_ISAC_iotj,
  author={Wei, Zhiqing and Qu, Hanyang and Wang, Yuan and Yuan, Xin and Wu, Huici and Du, Ying and Han, Kaifeng and Zhang, Ning and Feng, Zhiyong},
  journal=IEEE_Internet_of_Things_Journal, 
  title={Integrated Sensing and Communication Signals Toward {5G-A} and {6G}: A Survey}, 
  year={2023},
  volume={10},
  number={13},
  pages={11068-11092},
  doi={10.1109/JIOT.2023.3235618}}

@ARTICLE{lu_ISAC_iotj,
  author={Lu, Shihang and Liu, Fan and Li, Yunxin and Zhang, Kecheng and Huang, Hongjia and Zou, Jiaqi and Li, Xinyu and Dong, Yuxiang and Dong, Fuwang and Zhu, Jia and Xiong, Yifeng and Yuan, Weijie and Cui, Yuanhao and Hanzo, Lajos},
  journal=IEEE_Internet_of_Things_Journal, 
  title={Integrated Sensing and Communications: Recent Advances and Ten Open Challenges}, 
  year={2024},
  volume={11},
  number={11},
  pages={19094-19120},
  doi={10.1109/JIOT.2024.3361173}}

@ARTICLE{Liu2024UAVIoT,
  author={Liu, Zechen and Liu, Xin and Liu, Yuemin and Leung, Victor C. M. and Durrani, Tariq S.},
  journal=IEEE_Transactions_on_Wireless_Communications, 
  title={{UAV} Assisted Integrated Sensing and Communications for Internet of Things: {3D} Trajectory Optimization and Resource Allocation}, 
  year={2024},
  volume={23},
  number={8},
  pages={8654-8667},
  doi={10.1109/TWC.2024.3352985}}

@ARTICLE{qianglong_tutorial,
  author={Dai, Qianglong and Zeng, Yong and Wang, Huizhi and You, Changsheng and Zhou, Chao and Cheng, Hongqiang and Xu, Xiaoli and Jin, Shi and Lee Swindlehurst, A. and Eldar, Yonina C. and Schober, Robert and Zhang, Rui and You, Xiaohu},
  journal=IEEE_Communications_Surveys_Tutorials, 
  title={A Tutorial on {MIMO-OFDM ISAC}: From Far-Field to Near-Field}, 
  year={2026},
  month={Jan.},
  volume={28},
  number={},
  pages={4319-4358},
  doi={10.1109/COMST.2025.3650568}}

@ARTICLE{yuxuan_magazine,
  author={Song, Yuxuan and Zeng, Yong and Yang, Yuhang and Ren, Zixiang and Cheng, Gaoyuan and Xu, Xiaoli and Xu, Jie and Jin, Shi and Zhang, Rui},
  journal=IEEE_Communications_Magazine, 
  title={An Overview of Cellular {ISAC} for Low-Altitude {UAV}: New Opportunities and Challenges}, 
  year={2025},
  month={Dec.},
  volume={63},
  number={12},
  pages={88-95},
  doi={10.1109/MCOM.002.2400742}}

@ARTICLE{Yong_UAV,
  author={Zeng, Yongs and Wu, Qingqing and Zhang, Rui},
  journal={Proc. IEEE}, 
  title={Accessing From the Sky: A Tutorial on {UAV} Communications for 5{G} and Beyond}, 
  year={2019},
  month={Dec.},
  volume={107},
  number={12},
  pages={2327-2375},
  doi={10.1109/JPROC.2019.2952892}}

@ARTICLE{liu_isac_jsac2022,
  author={Liu, Fan and Cui, Yuanhao and Masouros, Christos and Xu, Jie and Han, Tony Xiao and Eldar, Yonina C. and Buzzi, Stefano},
  journal=IEEE_Journal_on_Selected_Areas_in_Communications, 
  title={Integrated Sensing and Communications: Toward Dual-Functional Wireless Networks for {6G} and Beyond}, 
  year={2022},
  volume={40},
  number={6},
  pages={1728-1767},
  doi={10.1109/JSAC.2022.3156632}}

@ARTICLE{liu_crb_tsp2022,
  author={Liu, Fan and Liu, Ya-Feng and Li, Ang and Masouros, Christos and Eldar, Yonina C.},
  journal=IEEE_Transactions_on_Signal_Processing, 
  title={Cramér-Rao Bound Optimization for Joint Radar-Communication Beamforming}, 
  year={2022},
  volume={70},
  number={},
  pages={240-253},
  doi={10.1109/TSP.2021.3135692}}

@ARTICLE{ren_crbrate_twc2024,
  author={Ren, Zixiang and Peng, Yunfei and Song, Xianxin and Fang, Yuan and Qiu, Ling and Liu, Liang and Ng, Derrick Wing Kwan and Xu, Jie},
  journal=IEEE_Transactions_on_Wireless_Communications, 
  title={Fundamental {CRB}-Rate Tradeoff in Multi-Antenna {ISAC} Systems With Information Multicasting and Multi-Target Sensing}, 
  year={2024},
  volume={23},
  number={4},
  pages={3870-3885},
  doi={10.1109/TWC.2023.3312723}}

@ARTICLE{hua_crbrate_twc2024,
  author={Hua, Haocheng and Han, Tony Xiao and Xu, Jie},
  journal=IEEE_Transactions_on_Wireless_Communications, 
  title={{MIMO} Integrated Sensing and Communication: {CRB}-Rate Tradeoff}, 
  year={2024},
  volume={23},
  number={4},
  pages={2839-2854},
  doi={10.1109/TWC.2023.3303326}}

@ARTICLE{lyu_uav_isac_twc2023,
  author={Lyu, Zhonghao and Zhu, Guangxu and Xu, Jie},
  journal=IEEE_Transactions_on_Wireless_Communications, 
  title={Joint Maneuver and Beamforming Design for {UAV}-Enabled Integrated Sensing and Communication}, 
  year={2023},
  volume={22},
  number={4},
  pages={2424-2440},
  doi={10.1109/TWC.2022.3211533}}

@ARTICLE{meng_uav_isac_mwc2024,
  author={Meng, Kaitao and Wu, Qingqing and Xu, Jie and Chen, Wen and Feng, Zhiyong and Schober, Robert and Swindlehurst, A. Lee},
  journal=IEEE_Wireless_Communications, 
  title={{UAV}-Enabled Integrated Sensing and Communication: Opportunities and Challenges}, 
  year={2024},
  volume={31},
  number={2},
  pages={97-104},
  doi={10.1109/MWC.131.2200442}}

@ARTICLE{jing_isac_sky_twc2024,
  author={Jing, Xiaoye and Liu, Fan and Masouros, Christos and Zeng, Yong},
  journal=IEEE_Transactions_on_Wireless_Communications, 
  title={{ISAC} From the Sky: {UAV} Trajectory Design for Joint Communication and Target Localization}, 
  year={2024},
  volume={23},
  number={10},
  pages={12857-12872},
  doi={10.1109/TWC.2024.3396571}}

@ARTICLE{jiang_uav_isac_cl2024,
  author={Jiang, Yifan and Wu, Qingqing and Chen, Wen and Meng, Kaitao},
  journal=IEEE_Communications_Letters, 
  title={{UAV}-Enabled Integrated Sensing and Communication: Tracking Design and Optimization}, 
  year={2024},
  volume={28},
  number={5},
  pages={1024-1028},
  doi={10.1109/LCOMM.2024.3379504}}

@ARTICLE{meng_uav_isac_twc_2023,
  author={Meng, Kaitao and Wu, Qingqing and Ma, Shaodan and Chen, Wen and Wang, Kunlun and Li, Jun},
  journal=IEEE_Transactions_on_Wireless_Communications, 
  title={Throughput Maximization for {UAV}-Enabled Integrated Periodic Sensing and Communication}, 
  year={2023},
  volume={22},
  number={1},
  pages={671-687},
  doi={10.1109/TWC.2022.3197623}}

@ARTICLE{Deng_beamforming_twc_2023,
  author={Deng, Cailian and Fang, Xuming and Wang, Xianbin},
  journal=IEEE_Transactions_on_Wireless_Communications, 
  title={Beamforming Design and Trajectory Optimization for {UAV}-Empowered Adaptable Integrated Sensing and Communication}, 
  year={2023},
  volume={22},
  number={11},
  pages={8512-8526},
  doi={10.1109/TWC.2023.3264523}}

@ARTICLE{Abdissa_deployment_precoder_iotj2024,
  author={Abdissa Bayessa, Gezahegn and Chai, Rong and Liang, Chengchao and Kumar Jain, Deepak and Chen, Qianbin},
  journal=IEEE_Internet_of_Things_Journal, 
  title={Joint {UAV} Deployment and Precoder Optimization for Multicasting and Target Sensing in {UAV}-Assisted {ISAC} Networks}, 
  year={2024},
  volume={11},
  number={20},
  pages={33392-33405},
  doi={10.1109/JIOT.2024.3430371}}

@ARTICLE{Xiong_tradeoff_tit2023,
  author={Xiong, Yifeng and Liu, Fan and Cui, Yuanhao and Yuan, Weijie and Han, Tony Xiao and Caire, Giuseppe},
  journal=IEEE_Transactions_on_Information_Theory, 
  title={On the Fundamental Tradeoff of Integrated Sensing and Communications Under Gaussian Channels}, 
  year={2023},
  volume={69},
  number={9},
  pages={5723-5751},
  doi={10.1109/TIT.2023.3284449}}

@ARTICLE{Hou_beamforming_jsac2024,
  author={Hou, Kaiyue and Zhang, Shuowen},
  journal=IEEE_Journal_on_Selected_Areas_in_Communications, 
  title={Optimal Beamforming for Secure Integrated Sensing and Communication Exploiting Target Location Distribution}, 
  year={2024},
  volume={42},
  number={11},
  pages={3125-3139},
  doi={10.1109/JSAC.2024.3431573}}

@ARTICLE{Mao_isacnet_jsac2026,
  author={Mao, Weihao and Lu, Yang and Pan, Gaofeng and An, Jianping and Ai, Bo and Ng, Derrick Wing Kwan},
  journal=IEEE_Journal_on_Selected_Areas_in_Communications, 
  title={Cramér–Rao Bound Optimization for Bistatic ISAC: Transceiver Design and Attention-Based ISACNet}, 
  year={2026},
  volume={44},
  number={},
  pages={181-195},
  doi={10.1109/JSAC.2025.3607866}}

@ARTICLE{Hua_transmit_beamform_tvt2023,
  author={Hua, Haocheng and Xu, Jie and Han, Tony Xiao},
  journal=IEEE_Transactions_on_Vehicular_Technology, 
  title={Optimal Transmit Beamforming for Integrated Sensing and Communication}, 
  year={2023},
  volume={72},
  number={8},
  pages={10588-10603},
  doi={10.1109/TVT.2023.3262513}}

@STRING{IEEE_Transactions_on_Vehicular_Technology         = "{IEEE} Trans. Veh. Technol."}

@STRING{IEEE_Journal_of_Selected_Topics_in_Signal_Processing       = "{IEEE} J. Sel. Topics Signal Process."}

@STRING{IEEE_Transactions_on_Signal_Processing         = "{IEEE} Trans. Signal Process."}

@STRING{IEEE_Communications_Letters       = "{IEEE} Commun. Lett."}

@STRING{IEEE_Journal_on_Selected_Areas_in_Communications       = "{IEEE} J. Sel. Areas Commun."}

@STRING{IEEE_Transactions_on_Wireless_Communications       = "{IEEE} Trans. Wireless Commun."}

@STRING{IEEE_Transactions_on_Information_Theory         = "{IEEE} Trans. Inf. Theory"}

@STRING{IEEE_Internet_of_Things_Journal        = "{IEEE} Internet Things J."}

@STRING{IEEE_Communications_Magazine        = "{IEEE} Commun. Mag."}

@STRING{IEEE_Communications_Surveys_Tutorials      = "{IEEE} Commun. Surveys Tuts."}

@STRING{IEEE_Network        = "{IEEE} Netw."}

@STRING{IEEE_Wireless_Communications         = "{IEEE} Wireless Commun."}

\vfill

\end{document}